\documentclass[a4paper,twocolumn,11pt]{quantumarticle}
\pdfoutput=1
\usepackage{hyperref}
\usepackage[utf8]{inputenc}
\usepackage[english]{babel}
\usepackage[T1]{fontenc}
\usepackage[utf8]{inputenc}
\usepackage[numbers, sort&compress]{natbib}
\usepackage{amsmath}
\usepackage{dsfont}
\usepackage{algorithm}
\usepackage{algpseudocode}
\usepackage{hyperref}
\usepackage{physics}
\usepackage{tikz}
\usepackage{lipsum}
\usepackage{tabularx}
\usepackage{algorithm}
\usepackage{algpseudocode}
\usepackage{multirow}
\usepackage{float}
\usepackage{titlesec}

\usepackage{amssymb}
\usepackage{qcircuit}
\usepackage{breqn}
\usepackage{wrapfig}

\newtheorem{theorem}{Theorem}
\newtheorem{proof}{Proof}
\newtheorem{corollary}{Corollary}
\newtheorem{definition}{Definition}
\begin{document}
\title{Extending matchgate simulation methods to universal quantum circuits}

\author{Avinash Mocherla}
\affiliation{%
Department of Physics and Astronomy, University College London, United Kingdom}
\author{Lingling Lao}
\affiliation{%
Department of Physics and Astronomy, University College London, United Kingdom}
\affiliation{%
School of Computer Science, Northwestern Polytechnical University, China}
\author{Dan E. Browne}
\affiliation{%
Department of Physics and Astronomy, University College London, United Kingdom}

\maketitle
% \begin{abstract}
% Matchgates are a family of parity-preserving two-qubit gates, nearest-neighbour circuits of which are known to be classically simulable in polynomial time. In this work, we present a simulation method to classically simulate an \boldsymbol{n}\boldsymbol{n}-qubit circuit containing \boldsymbol{N}\boldsymbol{N} gates, \boldsymbol{m}\boldsymbol{m} of which are universality-enabling gates and \boldsymbol{N-m}\boldsymbol{N-m} of which are matchgates, in the setting of single-qubit Pauli measurements and product state inputs. The universality-enabling gates we consider include the SWAP, CZ and CPhase gates. For fixed \boldsymbol{m}\boldsymbol{m} as \boldsymbol{n} \rightarrow \infty\boldsymbol{n} \rightarrow \infty, the resource cost, T, scales as \mathcal{O}(\frac{en}{m+1})^{2m+2}\mathcal{O}(\frac{en}{m+1})^{2m+2}. For \boldsymbol{m = \lambda n}\boldsymbol{m = \lambda n} for some,  \boldsymbol{T}\boldsymbol{T} scales as \mathcal{O}(2^{2nH(\frac{m+1}{n})})\mathcal{O}(2^{2nH(\frac{m+1}{n})}), where H(\lambda) = -\lambda \log_{2}(\lambda) - (1-\lambda) \log_{2}(1-\lambda)H(\lambda) = -\lambda \log_{2}(\lambda) - (1-\lambda) \log_{2}(1-\lambda) is the binary entropy function.

% \end{abstract}

\begin{abstract}
Matchgates are a family of parity-preserving two-qubit gates, nearest-neighbour circuits of which are known to be classically simulable in polynomial time. In this work, we present a simulation method to classically simulate an $\boldsymbol{n}$-qubit circuit containing $\boldsymbol{N}$ gates, $\boldsymbol{m}$ of which are universality-enabling gates and $\boldsymbol{N-m}$ of which are matchgates, in the setting of single-qubit Pauli measurements and product state inputs. The universality-enabling gates we consider include the SWAP, CZ, and CPhase gates. For fixed $\boldsymbol{m}$ as $\boldsymbol{n} \rightarrow \boldsymbol{\infty}$, the resource cost, $\boldsymbol{T}$, scales as $\boldsymbol{\mathcal{O}\left(\left(\frac{en}{m+1}\right)^{2m+2}\right)}$. For $\boldsymbol{m}$ scaling as a linear function of $\boldsymbol{n}$, however,  $\boldsymbol{T}$ scale as $\boldsymbol{\mathcal{O}\left(2^{2nH\left(\frac{m+1}{n}\right)}\right)}$, where $\boldsymbol{H}(\lambda)$ is the binary entropy function.
\end{abstract}

% = -\lambda \log_{2}(\lambda) - (1-\lambda) \log_{2}(1-\lambda)}$

In quantum computing context, matchgates refer to a group of two-qubit parity-preserving gates of the form: 
\[
G(A,B) = \begin{pmatrix}
a & 0 & 0& b\\
0 & e & f& 0\\
0 & g& h& 0\\
c & 0 & 0& d\\
\end{pmatrix} 
\]
\[
A = \begin{pmatrix}
a & b \\
c & d 
\end{pmatrix}, 
B = \begin{pmatrix}
e & f \\
g & h 
\end{pmatrix}
\]
where $det(A) = det(B) = \pm 1$. Gates of this type commonly occur in domains such as quantum chemistry \cite{arute2020observation}, fermionic linear optics \cite{knill2001fermionic}, quantum machine learning \cite{johri2021nearest}, and comprise part of the native gate set of several quantum computing architectures. An important fact about matchgates is that when they are composed into nearest neighbour circuits, they are efficiently classically simulable. This was first discovered by Valiant \cite{valiant2002quantum} in the context of `perfect matchings' of a graphical representation of matchgate unitaries. It was then extended to a fermionic context by Terhal and DiVincenzo \cite{terhal2002classical}, and later Josza and Miyake \cite{Josza2008}, each of whom developed classical algorithms to efficiently simulate matchgate circuits in different input and output regimes. When confronted with the fact that matchgates are efficiently simulable, a natural question is whether such efficiency is maintained as small numbers of a universality-enabling primitive outside of the matchgate group are added. For example, such primitives could include CZ, SWAP and CPhase gates, which are all examples of `parity preserving non-matchgates' \cite{Brod2011}, (for brevity we refer to these as ZZ gates). Is it possible to develop a method to simulate these universal matchgate circuits, dubbed `matchgate + ZZ' circuits?

To answer this question we develop a Pauli-basis simulation method, which is able to simulate an $n$-qubit `matchgate + ZZ' circuit containing $m$ ZZ gates, in time which is polynomial in $n$ when $m$ is fixed, and scales as $\mathcal{O}(2^{2nH(\frac{m+1}{n})})$ for $m \leq \frac{n}{2}$ in general, where $0 < H(\lambda) \leq 1$ is the binary entropy function. The use of a Pauli-basis method stems from the equivalence  of fermionic and Pauli operator descriptions of matchgates via the Jordan Wigner representation. This allows us to succinctly represent the `matchgate + ZZ' simulation problem in Liouville notation. In this notation, matchgate circuits manifest as block diagonal super-operators acting on closed (mostly poly-sized) linear spans of Pauli-operators. On the other hand, ZZ gates manifest as superoperators which act across several linear spans of Pauli operators. Given these two observations, classical simulation of matchgate + ZZ circuits is possible if, starting from a closed linear span, we adaptively keep track of the newly accessed Pauli basis elements as each ZZ gate is applied.

The paper is structured as follows. In section 1, we review the properties of matchgates relevant to classical simulation. We introduce methods for simulating matchgate circuits from \cite{Josza2008} and show how this formalism can be readily recast in Liouville notation. In Section 2, we extend these insights to non-matchgates and show how matchgates can in principle be simulated with any non-matchgate by considering the structure of the corresponding superoperators. In Section 3, we introduce the Pauli-basis simulation technique specifically for `matchgate + ZZ' circuits, and in Section 4, determine its scaling in the regimes of fixed $m$ and and variable $m$. Finally, in Section 5, we verify the asymptotic scalings by performing numerical simulations of `matchgate + ZZ' circuits arising in the Fermi-Hubbard model.

% To answer this question, we first study the `matchgate + ZZ' simulation problem in Liouville notation. In this picture, we can show that matchgate circuits manifest as block diagonal super-operators acting on closed (mostly poly-sized)  linear spans of Pauli-operators. This implies that if the dynamics are restricted to a single block, efficient classical simulation is possible. We then show how the action of ZZ gates is in general to connect these otherwise closed linear spans. 
% This allows us to develop a simple scheme for simulation of a universal circuit, by the adaptively increasing the dimension of the encompassing vector space as each ZZ gate is applied. To do this optimally, we develop a Pauli-basis simulation technique and derive asymptotic scalings for the algorithmic runtime for two different types of circuit structure, namely random circuits and layered circuits, which represent the worst-case and average-case simulation difficulty respectively. Finally, this extended scheme is demonstrated on circuits arising in the Fermi-Hubbard model, which are known to contain predominantly matchgates and a few CPhase gates.
To begin, we introduce the salient properties of matchgates relevant to our discussion. 

\section{Review of Matchgate Simulation} 
Matchgates, which in quantum computing are gates of the form $G(A,B)$, where $det(A) = det(B) = \pm 1$, are closely associated with the dynamics of non-interacting fermions. This link stems from the fact that matchgates are the result of mapping a subset of so-called \textit{fermionic Gaussian operations}, which arise in fermionic physics, to quantum computation. The fermionic setting allows us to understand why circuits of nearest-neighbour matchgates can be efficiently simulated and is the starting point of further analysis.

\subsection{Gaussian operations}
Consider a system of $n$ fermionic modes with the $k$th mode associated with a creation and annihilation operator, $a_{k}^{\dagger}$ and $a_{k}$ respectively. It has been shown that there exists a mapping between these fermionic operators and spin-$\frac{1}{2}$ operators (Pauli matrices) via the Jordan Wigner representation. To demonstrate this map succinctly, it is useful to define a set of $2n$ hermitian operators $c_{2k-1} = a_{k}^{\dagger} + a_{k}$ and $c_{2k} = -i(a_{k}^{\dagger} - a_{k})$, sometimes referred to as Majorana spinors, which pairwise satisfy the following anti-commutation relations for $\quad \mu, \nu=1, \ldots, 2 n$:
\begin{equation} 
\label{Equation: Anti-commutation}
    \left\{c_{\mu}, c_{\nu}\right\} \equiv c_{\mu} c_{\nu}+c_{\nu} c_{\mu}=2 \delta_{\mu \nu} I.
\end{equation}

Using the Jordan-Wigner representation, each spinor can be written in terms of strings of Pauli operators as follows:
\begin{equation} 
\begin{aligned} \label{Equation: Jordan-Wigner Map}
&c_{2k-1} =\left(\prod_{i=1}^{k-1} Z_{i}\right)X_{k} \\
&c_{2k}=\left(\prod_{i=1}^{k-1} Z_{i}\right)Y_{k}.
\end{aligned}
\end{equation} 

It is known that matchgates correspond to unitary operators generated from the set of nearest-neighbour operators $\mathcal{X} = \{X_{k}X_{k+1}, X_{k}Y_{k+1}, Y_{k}Y_{k+1}, Y_{k}X_{k+1}, Z_{k},Z_{k+1}\}$. Using equation \ref{Equation: Jordan-Wigner Map} we can see each of the elements of $\mathcal{X}$ can be expressed as quadratic Majorana monomials:
$$
Z_{k}= -ic_{2k-1}c_{2k},
$$
for $k \in [1, \ldots, n] $, and
$$
\begin{gathered}
X_{k} X_{k+1}= -ic_{2k}c_{2k+1} \\
Y_{k} Y_{k+1}= ic_{2k-1}c_{2k+2} \\
Y_{k} X_{k+1}= ic_{2k-1}c_{2k+1} \\
X_{k} Y_{k+1}= -ic_{2k}c_{2k+2}
\end{gathered}
$$

This highlights that a matchgate is the \textit{Gaussian operation}, $U_{MG} = e^{iH_{MG}}$ generated from a Hamiltonian $H_{MG}$ which is a linear combination of quadratic monomials:
\begin{equation} \label{Equation: Matchgate quadratic Hamiltonian}
H_{MG} = -i\sum^{2k+2}_{\mu \neq \nu = 2k-1} \alpha_{\mu\nu}c_{\mu}c_{\nu}, 
\end{equation} 

where $\alpha_{\mu\nu}$ is a real anti-symmetric matrix, and the spinors are drawn from the set $\{c_{2k-1},c_{2k},c_{2k+1},c_{2k+2}\}$. The set consists of spinors that correspond to nearest-neighbour modes, specifically $k$ and $k + 1$. However, Gaussian operations can also be generated from quadratic monomials connecting non-nearest neighbour modes, such as $c_{1}c_{5}$ or  $c_{2}c_{6}$, ie. from a more general Hamiltonian: 
\begin{equation} \label{Equation: General quadratic Hamiltonian}
H = -i\sum^{2n}_{\mu,\nu} \alpha_{\mu\nu}c_{\mu}c_{\nu}. 
\end{equation} 

Such monomials, expressed in terms of Pauli operators, act non-trivially across multiple qubits. It can be shown, however, that the Gaussian operation $e^{iH}$ generated by $H$ can be efficiently decomposed as a circuit of $\mathcal{O}(n^{3})$ nearest-neighbour matchgates \cite{Josza2008}, when expressed in terms of qubits. Hence we can consider a general Gaussian operation to correspond to a nearest-neighbour matchgate \textit{circuit}.

\begin{table*}[t] 

    \centering
    \begin{tabular}{|p{0.3cm}|p{5cm}|p{6cm}|}

    \hline
    $k$ & $ L_{2}^{(k)}$  & J-W-equivalent Paulis\\ 
    \hline
    $0$ &  $I$ &  $II$  \\
    $1$ &  $c_{1}, c_{2},c_{3},c_{4}$  &  $XI,YI,ZX,ZY$  \\
    $2$  & $ c_{1}c_{2}, c_{1}c_{3},c_{1}c_{4}, c_{2}c_{3},c_{2}c_{4},c_{3}c_{4} $  & $iZI, iYX, iYY, iXX, iXY, iIZ$\\
    $3$  &  $c_{1}c_{2}c_{3},c_{1}c_{2}c_{4},c_{1}c_{3}c_{4},c_{2}c_{3}c_{4}$  & $iIX,iIY,iXZ,iYZ$   \\
    $4$ &  $c_{1}c_{2}c_{3}c_{4}$  &  $-ZZ$  \\
    \hline
    \end{tabular}
    \caption{Spinor basis $L_{2}$ and Pauli basis $P_{2}$ for all values of $k \in [0,4]$.}
    \label{Table: Majorana Pauli Comparison}
    
\end{table*}

\begin{figure*}[t] 
    \centering
    \includegraphics[width=0.8\textwidth]{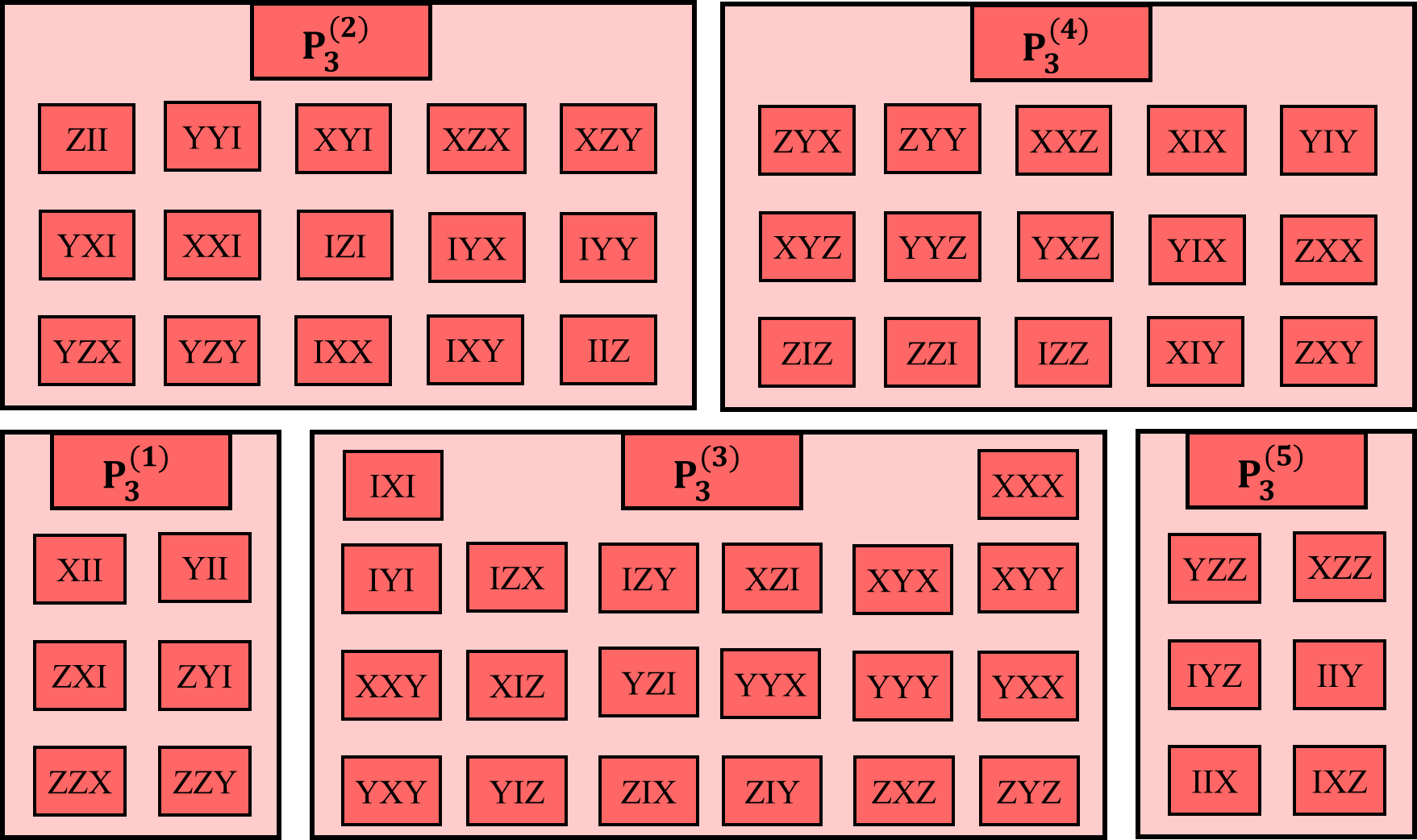}
    \caption{Subsets of $P_{3}$. There are 5 subsets, denoted as $P_{3}^{(k)}$ each of size ${6 \choose k}$. $P_{3}^{(0)} = III$ and $P_{3}^{(6)}=ZZZ$, as well as factors of $i$ have been omitted for clarity.}
    \label{Figure: linear_spans}
\end{figure*}

\subsection{PI-SO versus CI-MO simulation}

Nearest-neighbour circuits of matchgates are well-known for their efficient classical simulability in specific input-output settings. The two primary setting are known as the PI-SO (product-input, single-output) and CI-MO (computational input, multi-qubit output) settings, as coined by Brod in \cite{Brod2016}. In the PI-SO setting, computations involve nearest-neighbour matchgate circuits with product state inputs and measurements of single qubits. The simulation of such computations was established by Josza and Miyake \cite{Josza2008}, who exploited the Clifford algebraic properties of spinors (to be shown) to enclose the dynamics in a poly-sized vector space.

On the other hand, in the CI-MO setting, which focuses on computational basis states and measurements on subsets of output modes, it was demonstrated \cite{terhal2002classical} that the marginal probability p(y|x) for input and output bitstrings x and y can be efficiently evaluated. This was achieved by reformulating the computation as the evaluation of the Pfaffian of an antisymmetric matrix, using Wick's theorem. Although both settings allow for efficient classical simulation, the underlying mechanisms behind their efficiency are not evidently related.
Consequently, we direct our attention to the fundamental insights derived from the PI-SO setting and explore their applicability to universal circuits.

\subsection{Clifford Algebra Formalism}

While the spinors underlying the definition of matchgates have a natural fermionic interpretation, they can also be interpreted abstractly as the generators of a $2^{2n}$ dimensional Clifford algebra denoted $\mathcal{C}_{2n}$, as mentioned in \cite{Josza2008}. The relevance of this observation to quantum computation is that the basis of $\mathcal{C}_{2n}$:
\begin{multline} \label{equation: products of spinors}
L_{n} = \{I, c_{\mu}, c_{\mu\nu}=c_{\mu} c_{\nu},  \cdots, c_{12 \cdots 2n}=c_{1} c_{2} \cdots c_{2n} \\ \mid 1 \leq \mu<\nu < \cdots  \leq 2n \},
\end{multline}
econstructed from unique ascending order products of spinors $c_{\mu}$, forms a basis for $2^{n} \times 2^{n}$ hermitian matrices \cite{hestenes2012clifford} \cite{wan2022matchgate}. Interestingly, such a property is also realised by the $2^{2n}$ dimensional n-qubit Pauli basis $\mathcal{P}_{n} = \{I, X, Y, Z\}^{\otimes n}$, and indeed, using the Jordan-Wigner representation, each element of $L_{n}$ can be represented by an n-qubit tensor product of Pauli operators (up to a global phase). This is demonstrated in Table \ref{Table: Majorana Pauli Comparison} for n = 2. An important consequence of this mapping is that the Pauli basis elements inherit the same graded structure as the underlying Clifford algebra. To make this clear, it is useful to label subsets of $L_{n}$ by the degree of their constituent spinors, $k$, which we denote by $L_{n}^{(k)}$. The size of each of these sets is  ${2n\choose{k}}$. There are 2$n$+1 such sets (including $L_{n}^{(0)} = I$), and the linear span of the elements in $L_{n}^{(k)}$ is written as $\mathcal{L}_{n}^{(k)}$. We can see that as a vector space $\mathcal{C}_{2n} = \bigoplus_{k}\mathcal{L}_{n}^{(k)} $ and thus is graded. Each span $\mathcal{L}_{n}^{(k)}$ can be identified with a Pauli basis vector space $\mathcal{P}_{n}^{(k)}$ such that $\mathcal{P}_{n} = \bigoplus_{k}\mathcal{P}_{n}^{(k)}$. This structure is illustrated in Figure \ref{Figure: linear_spans} for $n = 3$.

% Clifford algebras are vector spaces equipped with a bilinear product called the 'wedge product', generated from linear combinations of products of elements of an underlying vector space V. For our purposes, V is the 2n dimension vector space of spinors, and the implications of the wedge product are neatly codified in the anti-commutation relations of equation 1. Specifically, it is known that the $2n$ spinors form a  2n dimensional vector space V. Together with the anti-commutation relations of equation 1,  generate a $2^{2n}$ dimensional Clifford algebra, denoted $\mathcal{C}_{2n}$. A clifford algebra is a  For our purposes, as we make no use of the bilinear operator, but rather the graded vector space,  the basis of this algebra denoted $L_{2n}$:

% \begin{multline} \label{equation: products of spinors}
% L_{2n} = \{I, c_{\mu}, c_{\mu\nu}=c_{\mu} c_{\nu},  \cdots, c_{12 \cdots 2n}=c_{1} c_{2} \cdots c_{2n} \\ \mid 1 \leq \mu<\nu < \cdots  \leq 2n \}
% \end{multline}
% constructed from all unique ascending order products of spinors $c_{\mu}$.

% Such an algebra has the important property of being graded

% Clifford algebra are known to be graded.  A consequence of this fact is that the $4^{n}$-dimensional basis $L_{n}$:  
% %

% constructed from all unique ascending order products of spinors $c_{\mu}$, can be related bijectively (up to overall phase factors) to the $n$-qubit Pauli basis via the Jordan-Wigner representation. 

\subsection{PI-SO Matchgate Simulation}
We now show how this graded structure is exploited for classical simulation. Let us consider a typical computation in the PI-SO setting as the evaluation of the following expectation value: 

\begin{equation} \label{Equation: PISO}
\langle Z_{j} \rangle =  \bra{\psi}\mathcal{U}^{\dagger}(Z_{j})\mathcal{U}\ket{\psi},
\end{equation}

where $Z_{j}$ is the $Z$ Pauli operator on the $\textit{j}th$ qubit, $\mathcal{U}$ is the unitary denoting matchgate circuit with $N \thicksim poly(n)$ gates, and $\ket{\psi}$ is a product state. For a generic unitary, the classical resources required for this computation a priori scales exponentially. For a matchgate circuit, however, we can exploit the following key property of Gaussian operations, as shown by Josza and Miyake, to evaluate this expression in time $poly(n)$:

\begin{theorem} \label{Theorem: Josza Proof}
\cite{Josza2008} Let $H$ be a quadratic Hamiltonian as in \ref{Equation: General quadratic Hamiltonian} and $\mathcal{U} = e^{iH}$ the corresponding Gaussian operation. Then for all $\mu$:
\begin{equation}
\mathcal{U}^{\dagger} c_{\mu} \mathcal{U}=\sum_{\nu=1}^{2 n} \mathcal{T}_{\mu \nu} c_{\nu}
\end{equation}
where the matrix $\mathcal{T} \in SO(2n)$, and we obtain all of $S O(2 n)$ in this way.

\end{theorem}
\begin{proof} Write $c_\mu$ as $c_\mu(0)$ and introduce $c_\mu(t)=V(t) c_\mu(0) V(t)^{\dagger}$ with $V(t)=e^{i H t}$. Then
$$
\frac{d c_\mu(t)}{d t}=i\left[H, c_\mu(t)\right]
$$
(with square brackets $[a, b]$ denoting the commutator $a b-b a$ ). But $\left[c_{\nu_1} c_{\nu_2}, c_\mu\right]=0$ if $\mu \neq \nu_1, \nu_2$ and $\left[c_\mu c_\nu, c_\mu\right]=-2 c_\nu$ (from Equation \ref{Equation: Anti-commutation})   so
$$
\frac{d c_{\mu}(t)}{d t}=\sum_\nu 4 h_{\mu \nu} c_{\nu}(t) 
$$
$$
\quad c_{\mu}(t)=\sum_{\nu} \mathcal{T}_{\mu \nu}(t) c_{\nu}(0).
$$
and the theorem follows by just setting t = 1. 
\end{proof}
Theorem 1 implies that conjugation by a Gaussian operation maps individual spinors to the closed real linear span of all spinors. Alternatively stated, conjugation by matchgate maps elements of $\mathcal{L}_{n}^{(1)}$ to $\mathcal{L}_{n}^{(1)}$.  This applies to both matchgate circuits, $\mathcal{U}$, as well as individual matchgates $U$, as both are Gaussian operations. In each case, we will refer to the equivalent linear operator in $SO(2n)$ as $\mathcal{T}$ and $T$ respectively, and in general $\mathcal{T} = \prod_{i}^{N}T^{(i)}$ where $T^{(i)}$ is the linear operator for the $ith$ gate in a circuit.  As a corollary of Theorem 1, we can see that any degree monomial of spinors will be mapped to a linear span of the same degree, or equivalently the linear span $\mathcal{L}_{n}^{(k)
}$ is preserved by conjugation of matchgates. Explicitly: 

\begin{corollary}
    Let $H$ be a quadratic Hamiltonian as in \ref{Equation: General quadratic Hamiltonian} and $\mathcal{U} = e^{iH}$ the corresponding Gaussian operation. Let $S$ be any (ascending) order subset of $[0,2n]$, and $c_{S}$ be the  monomial of spinors with indices in $S$. Then for all S with cardinality $\vert S\vert$:

\begin{equation}
\mathcal{U}^{\dagger} c_{S} \mathcal{U} = \sum_{S^{'} \in {[2n] \choose  \vert S \vert}} det (\mathcal{T}\vert_{ S ,S^{'}}) c_{S^{'}},
\end{equation}
where ${[ 2n ] \choose  k }$ denotes subsets of $[0,2n]$ with cardinality $k$, and  $\mathcal{T}\vert_{ S ,S^{'}}$ indicates a minor of $\mathcal{T}$ whose rows and columns are indexed by $S$ and $S^{'}$ respectively.

\end{corollary} 

The coefficients  $det (\mathcal{T}\vert_{ S ,S^{'}})$ follow from applying the anti-commutation relations of the spinors, and are in fact the elements of a \textit{compound matrix }$\mathcal{R}^{(\vert S \vert)}_{n} \in SO({[2n]\choose \vert S \vert})$ for the given cardinality.

To understand how these results allow for efficient classical simulation in the PI-SO setting, let us rewrite Equation \ref{Equation: PISO}  using  
the fact that $Z_{j} = -ic_{2j-1}c_{2j}$:
% $$
% \langle Z_{j} \rangle =  -i\bra{\psi}\mathcal{U}^{\dagger}c_{2j-1}\mathcal{U}\mathcal{U}^{\dagger}c_{2j}\mathcal{U}\ket{\psi},
% $$
$$
\langle Z_{j} \rangle =  -i\bra{\psi}\mathcal{U}^{\dagger}c_{2j-1}c_{2j}\mathcal{U}\ket{\psi}.
$$
Of immediate consequence is that the measurement operator $-ic_{2j-1}c_{2j} \in \mathcal{L}_{n}^{(2)}$. 
This suggests by conjugating the measurement operator, otherwise known as a `Heiseburg picture' approach, we can write: 
%
% \begin{equation} \label{Equation: PISO-matchgate}
% \langle Z_{j} \rangle = -i\sum_{\mu\nu}^{2n}\mathcal{T}_{2j-1 \mu}\mathcal{T}_{2j \nu} \bra{\psi}c_{\mu}c_{\nu}\ket{\psi}.
% \end{equation}
\begin{equation} \label{Equation: PISO-matchgate}
\langle Z_{j} \rangle = -i\sum_{S,S^{'} \in {[2n] \choose 2}}det (\mathcal{T}\vert_{ \{2j-1,2j\} ,S^{'}}) \bra{\psi} c_{S^{'}} \ket{\psi},
\end{equation}
which is an expression which can be evaluated in polynomial time. This is due to the fact the sum consists of $\thicksim \mathcal{O}(n^2)$ terms, and for each term the determinant of constant-sized minors can be calculated efficiently from $\mathcal{T}$. Similarly, each expectation in the sum is the product of $n$ single-qubit operator expectation values (as $\ket{\psi}$ is product state), which is also efficient. Hence by exploiting the fact that matchgates preserve the graded structure of the operator space, operators like $-ic_{2j-1}c_{2j}$ which start in a poly-sized vector space will remain in a poly-sized vector space.

% Expanding on corollary 1, we can reinterpret Equation 8 within the vector space of spinors. Conjugation in this context involves applying a poly-sized degree two compound matrix to a vector that contains the coefficients of degree two spinors, ordered lexicographically. This reinterpretation allows us to generalize the simulation procedure to include non-matchgates. To formalize this ntion, we introduce the use of Liouville notation in the following discussion.

% Equivalently using Corollary 1, we can write this as 

% \begin{equation}
% \begin{split}
% \langle Z_{j} \rangle = \sum_{S^{'} \in {[2n] \vert \choose  2}} det            (\mathcal{T}\vert_{ S ,S^{'}}) \bra{\psi}c_{S^{'}}\ket{\psi}, \
% =\sum_{S^{'} \in {[2n] \vert \choose  2}} det            (\mathcal{T}\vert_{ S ,S^{'}}) \bra{\psi}c_{S^{'}}\ket{\psi},
% \end{split}
% \end{equation}

% Hence, we can conclude the overarching linear operator corresponding to a matchgate circuit, which we label $R_{MG}$, has a block-diagonal structure where each block is a compound matrix of size ${2n\choose \vert S \vert}$.

\subsection{Liouville notation}
On the contrary, non-matchgates do not preserve the graded structure of the clifford algebra of spinors. In general, they act across several linear spans. To capture this behaviour, we recast the problem in the $4^{n}$ dimensional Pauli-basis via Liouville notation. Specifically,  we can express a state $\rho$ as a real column vector, which we refer to as an operator vector, $|\rho\rangle\rangle \in \mathbb{R}^{4^{n}}$,
$$
|\rho\rangle\rangle=\left[\begin{array}{lll}
\cdots & \rho_{\sigma} & \cdots
\end{array}\right]^{T}
$$
where each vector element is the coefficient of the corresponding operator in its Pauli basis decomposition. That is,
$$
\rho_{\sigma}=\operatorname{Tr}(\sigma \rho)
$$
Here, $\sigma \in P_{n} $
is an $\textit{n}$-fold product of Pauli operators. We can also express an observable (i.e. Hermitian operator) $M$ as a real row vector:
$$
\langle\langle M|= \left[\begin{array}{lll}
\cdots & M_{\sigma} & \cdots
\end{array} \right],
$$
where each vector element is given by:
$$
M_{\sigma}=\operatorname{Tr}(\sigma M).
$$
 
In this notation, the action of conjugation on the operator vector $|\rho\rangle\rangle$ is now given by a linear operator $R \in SO(4^{n})$ acting on $|\rho\rangle\rangle$ denoted $|R\vert\rho\rangle\rangle$. Furthermore, the expectation value of an observable $M$ can be written:
\begin{equation} \label{Equation: Pauli-basis computation}
\langle M \rangle = \langle\langle M |R \vert\rho_{0} \rangle\rangle,
\end{equation} which is equivalent to $\operatorname{Tr}(M U^{\dagger} \rho_{0}U)$ for a circuit $U$. Given that the objects in equation \ref{Equation: Pauli-basis computation} are exponentially sized, a natural question is how such an expression can be evaluated efficiently. To do this, we introduce Algorithm \ref{Algorithm 1} [see appendix A]. Algorithm \ref{Algorithm 1} uses a sparse representation of each object to simulate the entire circuit in time which scales linearly in a quantity called $\chi_{t}$. We define this below.

\begin{definition} \label{definition: Pauli rank}
The \textbf{Pauli rank} of $\vert \rho \rangle\rangle$, denoted $\chi(\rho)$,  is the number of non-zero coefficients in its Pauli-basis decomposition. 
\end{definition}

\begin{definition} \label{definition: Total Pauli rank}
The \textbf{Total Pauli rank} of a \textit{circuit}consisting of $N$ gates, denoted $\chi_{t}$, is the sum of the Pauli ranks obtained after applying each gate in the circuit. That is: 
\begin{equation} \label{Total Pauli Rank}
        \chi_{t} = \sum_{i}^{N}\chi(\rho_{i}).
\end{equation}
\end{definition}
Using Algorithm 1 on a matchgate circuit in the PI-SO setting, for example, we see that the maximum value of the Pauli rank (using a Heisenberg approach to matrix-vector multiplication) is $\chi = {2n \choose 2}$. The total rank is, therefore, $\chi_{t} = N\chi$, which implies the scaling is the same as shown by Josza and Miyake in \cite{Josza2008}. We will also make use of Algorithm 1 in evaluating the more general 'matchgate + ZZ' circuits introduced in the following sections.

\section{Extension to non-matchgates}
We now consider how non-matchgates manifest in the context of the graded structure of the Pauli-operator basis. To do this, we consider how non-Gaussian operations (ie. those generated by $H \in \mathcal{L}_{n}^{(d)}$ with $d \neq 2$) transform elements of a linear span $\mathcal{L}_{n}^{(k)}$ under conjugation. From this, we can deduce the structure of the corresponding super-operator $R:$ $\mathcal{C}_{2n} \rightarrow \mathcal{C}_{2n}$. What  we find is that the application of such an operation induce transformations across \textit{multiple}  linear spans. This leads to a predictable increase in the Pauli rank of the measurement operator vector, allowing us to quantify the computational cost of simulating non-Gaussian operations.  To relate non-Gaussian operations to unitary gates, we show that any single and two qubit gate: $U^{(1)} \in U(2)$ and $U^{(2)} \in U(4)$ respectively, can be reduced to two specific types of non-Gaussian operation. Namely, those operations generate odd-degree ad quartic elements of $L_{n}$ respectively. We can hence characterise the classical cost of simulating any gate of practical relevance in terms of its effect on the Pauli rank of a simulation. 

\subsection{Transformations induced by non-Gaussian operations}

To start, we note we can simplify our notion of non-Gaussian operations to those generated from  $ H \in L_{n}^{(k)}$ , rather than $ H^{'} \in \mathcal{L}_{n}^{(k)}$, without loss of generality. This is from the observation that we can generate any element $H^{'}$  as $U^{\dagger}HU$, using Corollary 1 with a specially selected $\mathcal{T}$. The non-Gaussian operation generated from this new Hamiltonian, given by  $e^{iU^{\dagger}HU}$ is equal to $U^{\dagger}e^{iH}U$, which implies that all non-Gaussian operations generated from $\mathcal{L}^{(k)}_{n}$ are equivalent, up to conjugation by a matchgate, to the non-gaussian operation generated from the set $L_{n}^{(k)}.$ We refer to this property as matchgate-equivalence.

Let us consider the transformation of a basis element $c_{\{k\}} \in L_{n}^{(k)}$, where through abuse of notation the subscript ${k}$ refers to an (ascending order) subset of [0,2n] with cardinality $k$. We specifically consider the transformation $V^{\dagger}c_{\{k\}}V$, where $V$ is a non-Gaussian unitary of the form $e^{i \frac{\theta}{2} c_{\{d\}}}$, and $c_{\{d\}} \in L_{n}^{(d)}$. Now, we have:
 % We further constrain $\theta \neq j\pi$ for $j \in \mathbb{Z}$ to guarantee non-trivial behaviour
\begin{equation} \label{Equation: Non-Gaussian Conjugation}
\begin{split}
V^{\dagger}c_{\{k\}}V = 
& \cos^{2}{(\frac{\theta}{2})}c_{\{k\}} +\sin^{2}{(\frac{\theta}{2})}c_{\{d\}}c_{\{k\}}c_{\{d\}}+  \\
&  i\sin(\frac{\theta}{2})\cos(\frac{\theta}{2}) (c_{\{d\}}c_{\{k\}} - c_{\{k\}}c_{\{d\}}),  \\
\end{split}
\end{equation}

where it can be seen that the expression simplifies depending on whether the monomials $c_{\{d\}}$ and $c_{\{k\}}$ commute or anti-commute. Indeed, this depends on the number of sign flips induced by rearranging $c_{\{d\}}c_{\{k\}} \rightarrow c_{\{k\}}c_{\{d\}}$, which is related to the values $d$, $k$ and a quantity $l \in [0,\textrm{min}(k,d)]$, which we define as $l = |\{k\} \cap \{d\}|$ (i.e. the number of spinors which occur in both  $c_{\{d\}}$ and $c_{\{k\}}$):
\newtheorem{lemma}[theorem]{Lemma}

\begin{lemma} 
Let $k$ and $d$ be the degrees of $c_{\{k\}}$ and $c_{\{d\}}$ respectively.
Let $l = |\{k\} \cap \{d\}|$ be the number of spinor indices common to both. Then
$c_{\{d\}}c_{\{k\}}$ = $(-1)^{dk - l}c_{\{k\}}c_{\{d\}}$.
\end{lemma}

\begin{proof}
Majorana spinors anti-commute under permutation. This induces a sign depending on the number of permutations needed to transform $c_{\{k\}}c_{\{d\}} \rightarrow c_{\{d\}}c_{\{k\}}$. In the case l = 0 (no common spinors), the sign after this transform is $(-1)^{dk}$, as every spinor in $c_{\{k\}}$ is permuted past every spinor in $c_{\{d\}}$. When l is non-zero, l permutations will induce no sign (as these spinors will commute). Hence the overall sign induced after the full permutation is $(-1)^{dk - l}$.
\end{proof}

We can discern that in the commuting case ($dk - l$ is even), the transformation simplifies to the identity channel. However, in the anti-commuting case ($dk - l$ is odd), the following rotation occurs: 
\begin{equation}
\label{Equation: Non-Gaussian rotation}
V^{\dagger}c_{\{k\}}V  = \cos(\theta)c_{\{k\}} + i\sin(\theta)c_{\{d\}}c_{\{k\}},
\end{equation}
From equation \ref{Equation: Non-Gaussian rotation}, we see that non-matchgates transform $c_{\{k\}}$ to a space spanned by both $c_{\{k\}}$ and a monomial of the form $c_{\{d\}}c_{\{k\}}$. For a given value of $l$, $c_{\{d\}}c_{\{k\}}$ will simplify to $c_{\{d + k - 2l\}}$, as  two spinors with the same index will square to the identity.
With these results, we can specify the linear span of the transformed monomial $V^{\dagger}c_{\{k\}}V$ as follows, depending on the values of $d$, $k$ and $l$: 
% \clearpage
\begin{table}[h]
    % \hspace{-5mm}% 
    \setlength\tabcolsep{6pt}
    \begin{tabular}{|l|c|c|}
    % {0.5\textwidth}
    % {
    %   | >{\raggedright\arraybackslash}X 
    %   | >{\centering\arraybackslash}X 
    %   | >{\raggedleft\arraybackslash}X | }

     \hline
     $l$ parity & $kd$ odd & $kd$ even \\
     \hline
     $l \ even$  &  $\mathcal{L}_{n}^{(k)} \oplus \mathcal{L}_{n}^{(k+d -2l)}$  &  $\mathcal{L}_{n}^{(k)}$   \\
    \hline
     $l \ odd$ & $\mathcal{L}_{n}^{(k)}$   &  $\mathcal{L}_{n}^{(k)} \oplus \mathcal{L}_{n}^{(k+d -2l)}$     \\
     \hline
    \end{tabular}  \\
    \caption{Lookup table for the linear span to which $c_{\{k\}} \in L_{n}^{(k)}$ is transformed by a non-Gaussian operator of degree $d$, depending on the parity of $l$.}
    \label{Table: Non-Gaussian}
\end{table}

For a given $d$, $L_{n}^{(k)}$ will be split into $l$ subsets, each of which will be transformed differently by a non-Gaussian operation $V$, according to Table \ref{Table: Non-Gaussian}.
For an element of a \textit{linear span} $\mathcal{L}_{n}^{(k)}$, which is a linear combination of basis elements $c_{\{k\}} \in L_{n}^{(k)}$, we can say an element of $\mathcal{L}_{n}^{(k)}$ is transformed to the direct sum of the resultant linear spans of each transformed subset. To illustrate this, let us consider non-Gaussian operations corresponding to single- and two-qubit gates. 

\subsection{One-qubit non-Gaussian operations}

\begin{figure}[ht] 
    \centering\includegraphics[width=0.48\textwidth]{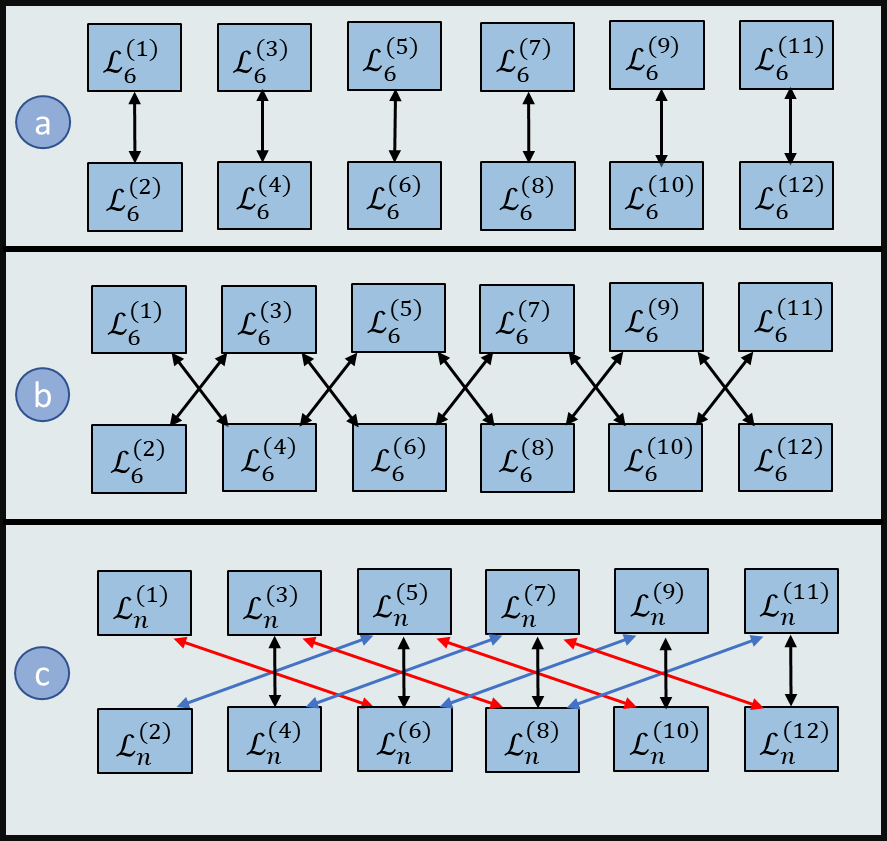}
    \caption{Transformation of an element of $\mathcal{L}_{6}^{(k)}$ under conjugation by an odd-degree non-Gaussian operation. Each arrow indicates the linear span to which a subset of $\mathcal{L}_{n}^{(k)}$ is transformed, corresponding to the possible values of $l$ for that basis set.   
    (a) $d = 1$,  (b) $d = 3$, (c) $d = 5$ (Coloured arrows are for visualisation only).}
    \label{Figure: odd transform}
\end{figure}

Consider that any single-qubit gate can be decomposed, up to a global phase, as the product of three Euler rotations: $U^{(1)} = R_{z}(\theta_{1})R_{y}(\theta_{2})R_{z}(\theta_{3})$, parameterised by angles $\theta_{1},\theta_{2},\theta_{3}$. The $R_{z}(\theta)$ rotations in this decomposition  can be expressed as matchgates, because  $G(R_{z}(\theta),R_{z}(\theta) )= R_z(\theta)\otimes \mathds{1}$, which implies that the non-Gaussian component of a single-qubit rotation is contained in $R_{y}(\theta)$. We can say that single-qubit gates $U^{(1)}$ are \textit{matchgate equivalent} to $R_{y}$ gates.

An $R_{y}$ rotation applied to qubit $j$  can be expressed as $\cos(\frac{\theta}{2})I - i\sin(\frac{\theta}{2})Y_{j}$, where $Y_{j} = (-i)^{j-1} \prod_{i}^{j-1}({c_{2i-1}c_{2i}})c_{2j}$. It can be seen that the $R_{y}$ rotation is generated by \textit{odd}-degree fermionic operators. Consider the simplest case where $j = 1$, which corresponds to a non-Gaussian operation of degree $d = 1$. The transformation induced on an element of $\mathcal{L}_{n}^{(k)}$ is given the following Theorem:

\begin{theorem}

\label{Theorem: Odd Conjugation}

Let $V = e^{i \theta c_{\{1\}}}$ be a degree 1 non-Gaussian operation. Then conjugation of an element of $\mathcal{L}_{n}^{(k)}$ by V is realized by a linear map $R_{V}^{(k)}$:

\begin{itemize}
    \item For $k \ odd$: \\ 
    $R_{V}^{(k)}:$ $ \mathcal{L}_{n}^{(k)} \oplus \mathcal{L}_{n}^{(k+1)} \rightarrow  \mathcal{L}_{n}^{(k)} \oplus \mathcal{L}_{n}^{(k+1)}  $
   
    \item For $k \ even$: \\
    $R_{V}^{(k)}:$  $ \mathcal{L}_{n}^{(k)} \oplus \mathcal{L}_{n}^{(k-1)}  \rightarrow  \mathcal{L}_{n}^{(k)} \oplus \mathcal{L}_{n}^{(k-1)}  $

\end{itemize}
\end{theorem}

\begin{proof}
For $d = 1$, the possible values of $l$ are 0 and 1. This implies each set $L_{n}^{(k)}$ is composed of two subsets, each of which is transformed according to Table 2. Specifically, for $d=1$,
 \vspace{11pt}
 
\hspace{-5mm}%    
\setlength\tabcolsep{12pt}
    \begin{tabular}{|l|c|c|}
 \hline
d = 1 & k odd & k even \\
 \hline
 $l = 0$  & $\mathcal{L}_{n}^{(k)} \oplus \mathcal{L}_{n}^{(k+1)}$  &   $\mathcal{L}_{n}^{(k)}$   \\
\hline
 $l = 1$ & $\mathcal{L}_{n}^{(k)}$  &     $\mathcal{L}_{n}^{(k)} \oplus \mathcal{L}_{n}^{(k-1)}$  \\
 \hline
\end{tabular} 

 \vspace{11pt}

When k is odd, it can be seen that for each value of $l$, an element of $\mathcal{L}_{n}^{(k)}$ is transformed to the linear span $\mathcal{L}^{(k)} \oplus \mathcal{L}_{n}^{(k+1)}$. Similarly, when k is even, an element of $\mathcal{L}_{n}^{(k)}$ is transformed to the linear span $\mathcal{L}_{n}^{(k)} \oplus \mathcal{L}^{(k-1)}_{n}$.
\end{proof}

The  structure of the linear operator $R_{V}: \mathcal{L}_{n} \rightarrow \mathcal{L}_{n}$, with the transformation due to each $R_{V}^{(k)}$ is shown graphically in figure \ref{Figure: odd transform}a. Here each arrow represents the transformation of a subset of $L_{n}^{(k)}$ corresponding to a particular value of $l$ (where a subset is transformed to its own linear span, no arrow is shown). It can be seen that under conjugation by a degree one non-Gaussian operator (which is matchgate-equivalent to a single-qubit gate on the first qubit),  Theorem \ref{Theorem: Odd Conjugation} corroborates a result in \cite{Josza2009} which states that matchgates plus arbitrary single-qubit gates on the first qubit are efficiently classically simulable in the PI-SO setting. This can be seen by the fact that the measurement operator vector will contain contain terms in $\mathcal{L}_{n}^{(1)} \oplus \mathcal{L}_{n}^{(2)}$, as single-qubit gates will, under conjugation, induce transformations in the space  $\mathcal{L}_{n}^{(1)} \oplus \mathcal{L}_{n}^{(2)}$. Separately, matchgates will induce transformations within the linear spans $\mathcal{L}_{n}^{(1)}$ and $\mathcal{L}_{n}^{(2)}$, which means that the overall Pauli rank will be bounded as ${2n \choose 1} + {2n \choose 2}$.

For odd-degree non-Gaussian operations corresponding to single qubit gates on higher indexed qubits, the transformations induced become more complicated, as shown in Figure \ref{Figure: odd transform}b ($d=3$) and Figure \ref{Figure: odd transform}c ($d=5$). In these cases, efficient classical simulation is no longer possible as applying multiple single qubit gates interspersed with matchgates will in general induce transformation in the direct sum of linear spans which encompass an exponentially scaling portion of $L_{n}$.

\subsection{Two-qubit non-Gaussian operations}

To understand the transformations induced by two qubit gates, we make use of the fact that $U^{(2)} \in U(4)$ can be written (via a KAK decomposition \cite{kraus2001optimal}\cite{Brod2011}), in the following form: 
$$
\begin{aligned}
% V &=\left(U_{1} \otimes U_{2}\right) U_{}\left(U_{3} \otimes U_{4}\right)  
U^{(2)} =\left(U_{1} \otimes U_{2}\right) e^{i(a X \otimes X+b Y \otimes Y+c Z \otimes Z)}\left(U_{3} \otimes U_{4}\right)
\end{aligned}.
$$
Here the parameters $a$, $b$ and $c$ satisfy $a\geq b \geq c$ (to ensure the decomposition is unique.) Comparing this to the KAK decomposition for an arbitrary matchgate:

\begin{equation*}
 \left(e^{i \phi_{1} Z} \otimes e^{i \phi_{2} Z}\right) e^{i(a X \otimes X+b Y \otimes Y)}\left(e^{i \phi_{3} Z} \otimes e^{i \phi_{4} Z}\right)
\end{equation*}

Comparing the two expressions, it can be seen that apart from the single qubit gates applied before and after the non-local core (which are matchgate-equivalent to $R_{y}$ rotations), by elimination, the gate $U_{ZZ}(\theta)$ = $e^{i \theta ZZ}$  must also induce non-Gaussian dynamics. Indeed, $U_{ZZ}(\theta)$ is a two-qubit parity-preserving gate which has been identified as enabling universal quantum computation when combined by Gaussian operations in  \cite{bravyi2002fermionic}, and corresponds to the unitary evolution of a degree four product of Majorana spinors. Specifically, when acting on qubits $a$ and $b$ (not necessarily nearest neighbour), one has $e^{i\theta Z_{a}Z_{b}} = e^{i \theta c_{2a-1},c_{2a},c_{2b-1},c_{2b}}$. It induces the following transformation on elements of $\mathcal{L}^{(k)}_{n}$: 
\begin{theorem} \label{Theorem: Even Conjugation}
Let $U_{ZZ}(\theta) = e^{i\theta Z_{a}Z_{b}} = e^{i\theta c_{2a-1}c_{2a}c_{2b-1}c_{2b}}$ be a non-Gaussian gate acting on qubits a,b. Then for $n \geq 3$, conjugation of an element of $\mathcal{L}_{n}^{(k)}$ by $U_{ZZ}$ is realized by a linear map $R_{ZZ}^{(k)}$: 
\begin{itemize}

    \item For $k \in [1,2]: \\
   R_{ZZ}^{(k)}:\mathcal{L}^{(k)}_{n}  \oplus \mathcal{L}^{(k+2)}_{n}  \rightarrow \mathcal{L}^{(k)}_{n}  \oplus \mathcal{L}^{(k+2)}_{n} $
    \item For $k \in [3,2n-2]: \\
     R_{ZZ}^{(k)}:\mathcal{L}^{(k-2)}_{n}  \oplus \mathcal{L}^{(k)}_{n} \oplus \mathcal{L}^{(k+2)}_{n} \rightarrow \\  
    \mathcal{L}^{(k-2)}_{n}  \oplus \mathcal{L}^{(k)}_{n} \oplus \mathcal{L}^{(k+2)}_{n}$
    \item for $k \in [2n-1,2n]: \\ R_{ZZ}^{(k)}: 
    \mathcal{L}^{(k)}  \oplus \mathcal{L}^{(k-2)}_{n} \rightarrow \mathcal{L}^{(k)}  \oplus \mathcal{L}^{(k-2)}_{n} $
    \label{Theorem: 4} 
\end{itemize}
\end{theorem}

\begin{proof}
    
For d = 4, the possible values of l are 0,1,2,3,4. Hence the set $L_{n}^{(k)}$ consists of five subsets for each value of l. Given the product 4k is always even, we need only consider the even column of Table 2, for which we enumerate the following transformations:
 \vspace{11pt}
 
\begin{tabularx}{0.47\textwidth} {
  | >{\raggedright\arraybackslash}X 
  | >{\centering\arraybackslash}X 
  | >{\raggedleft\arraybackslash}X | }
 \hline
 $ d = 4 $ &  4k even \\
 \hline
$l = 0$  &   $\mathcal{L}_{n}^{(k)}$   \\
\hline
$l = 1$ &     $\mathcal{L}_{n}^{(k)} \oplus \mathcal{L}_{n}^{(k+2)}$  \\
\hline
$l = 2$  &  $\mathcal{L}_{n}^{(k)}$   \\
\hline
$l = 3$ &    $\mathcal{L}_{n}^{(k)} \oplus \mathcal{L}^{(k-2)}_{n}$  \\
\hline
$l = 4$ &    $\mathcal{L}_{n}^{(k)}$  \\
 \hline
\end{tabularx} 
\vspace{11pt}

For the first case where k = 1 or 2, $l$ is limited to be 0,1, or 2. For any of these values, $\mathcal{L}^{(k)}_{n}$ is transformed to the linear span $\mathcal{L}^{(k)}_{n} \oplus \mathcal{L}^{(k+2)}_{n}$. 
 For k = 2n-1, 2n-2: $l$ can be either l = 4, 3 or 2, which means $\mathcal{L}^{(k)}_{n}$ is transformed to $\mathcal{L}^{(k)} \oplus \mathcal{L}^{(k-2)}_{n}$. Finally, in the general case where l can be 0,1,2,3 or 4, ${L}^{(k)}_{n}$ is transformed to $\mathcal{L}^{(k-2)} \oplus \mathcal{L}_{n}^{(k)} \oplus \mathcal{L}_{n}^{(k+2)}$.
\end{proof}

\begin{figure}[ht] 
    \centering
    \includegraphics[width=0.48\textwidth]{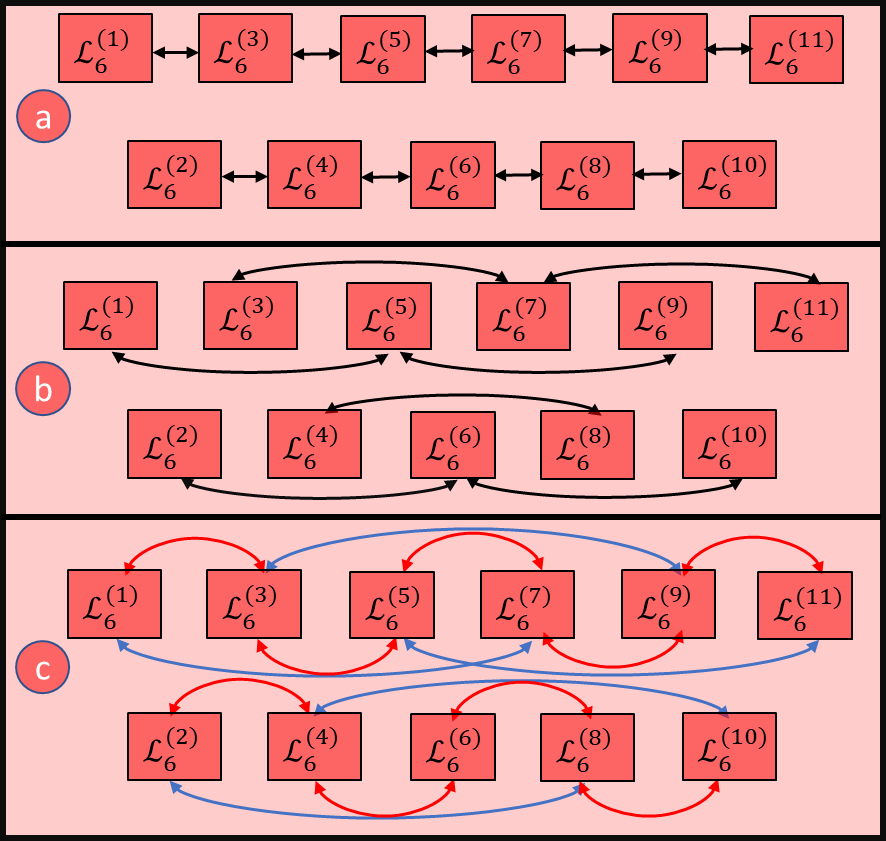}
    \caption{Transformation of an element of $\mathcal{L}_{6}^{(k)}$ under conjugation by an even-degree non-Gaussian operation. Each arrow indicates the linear span to which a subset of $\mathcal{L}_{n}^{(k)}$ is transformed, corresponding to the possible values of $l$ for that basis set. (a) $d = 4$,  (b) $d = 6$, (c) $d = 8$ (Coloured arrows are for visualisation only).}
    \label{Figure: even_transform}
\end{figure}

A schematic of the overall operation $R_{ZZ}: \mathcal{L}_{n} \rightarrow \mathcal{L}_{n}$ is shown in Figure \ref{Figure: even_transform}a. We can see that starting with a measurement operator vector in a given linear span, the action of conjugation, via the linear operator $R_{ZZ}$ will in general access adjacent equal parity spans. This gives the scheme to simulate ZZ gates - keep track of the new basis elements which are accessed as each $R_{ZZ}$ gate is applied, and update them as required. Such a method could be applied to any gate which is matchgate equivalent to $U_{ZZ}$ including SWAP, CZ, and CPhase.

As well as these, the transformations induced by higher degree non-Gaussian operations, specifically d = 6 and d = 8 are shown in Figure 3 for completeness. In general, these  would correspond to three and four-qubit gates in the Jordan-Wigner mapping and hence are not necessarily of practical relevance. 

% In this section, we have shown how elements of linear spans  $\mathcal{L}_{n}^{(k)}$, and hence $\mathcal{P}_{n}^{(k)}$, are transformed by the conjugation of an arbitrary non-Gaussian operation. Unlike Gaussian operations, which induce a transformation across a single span, $\mathcal{L}_{n}^{(k)}$, non-Gaussian operations induce linear maps across multiple spans. While we could consider simulation of matchgates with any non-matchgate by considering the relevant transformations in Figures 2 and 3 - in the following section, we focus on parity-preserving non-matchgates, dubbed `ZZ' gates, as firstly, they include gates such as SWAP, CZ and CPhase, which occur commonly in universal matchgate circuits. Secondly, from Figure 2a, we notice that applying an $R_{ZZ}$ operator connects adjacent linear spans in a predictable way. This gives a simple simulation procedure to simulate `matchgate + ZZ` gates which we show in the following section

\section{Classical Simulation of Matchgate + ZZ circuits} 
In the previous section, we showed how it is possible to simulate the action of a parity-preserving non-matchgate (a gate which is matchgate equivalent to $U_{ZZ}$) via the superoperator $R_{ZZ}$. Let us consider how adding such a gates affects the Pauli rank, and therefore the difficult of simulation. We denote an n-qubit `matchgate + ZZ' circuit containing $N$ gates, of which $m$ are parity-preserving non-matchgates and $N-m$ are matchgates, as $U_{MG+ZZ}$. This can be written in the following generic form: 
\begin{equation}
U_{MG + ZZ} = U_{MG}^{(m)}U_{ZZ}^{(m)}...U_{MG}^{(1)}U_{ZZ}^{(1)}U_{MG}^{(0)},
\end{equation}
where there are $m + 1$ nearest-neighbour matchgate circuits $U_{MG}$ interspersed with $m$ $U_{ZZ}$ gates.

A computation with this circuit in the PI-SO setting can then be written in Liouville notation as follows:

\begin{equation} \label{eqn:liouville}
\langle Z_{j} \rangle = \langle \langle Z_{j} \vert R_{MG+ZZ}\vert \rho_{0} \rangle \rangle. 
\end{equation}

where $\rho_{0}$ is a product state and $R_{MG+ZZ} = R_{MG}^{(m)}R_{ZZ}^{(m)}R_{MG}^{(m-1)}...R_{ZZ}^{(1)}R_{MG}^{(0)}$ are the corresponding superoperators. As has been mentioned, it is possible to evaluate an expression of the form of equation \ref{eqn:liouville} using Algorithm 1. This will have a scaling of $\mathcal{O}(\chi_{t})$. Hence, to determine the overall scaling of simulating a particular `matchgate + ZZ' circuit, we must find suitable bounds on $\chi_{t}$. This is the subject of the next section.

\section{Bounds on simulation time costs}

% In this section, we derive two asymptotic upper bounds for $\chi_{t}(N,n,m)$ corresponding to a `matchgate + ZZ' circuit $\mathcal{C}(N,n,m)$,  for when (i) $m$ is fixed and (ii) $m = \alpha n$. 

% To do this let us introduce the quantity $\chi_{max}$, which is the largest Pauli rank of the measurement operator at any point in the circuit. In terms of this quantity, we can loosely upper bound $\chi_{t}$ as $N\chi_{max}$. If $m$ ZZ gates are applied, then we expect that starting from $ZII \in \mathcal{P}_{n}^{(2)}$

% These bounds are generally applicable and depend solely on the values of $N$,$m$ and $n$. The second two bounds we derive consider the same asymptotic regimes as before, but make the additional assumption that the structure of $\mathcal{C}$ is \textit{layered}, in the sense that the ZZ gates are placed at regular intervals. We refer to each of these settings as \textit{general} and \textit{layered} respectively. 

% \begin{figure}[ht]
   
%     \centering
%     \hspace*{-0.7cm}
%     \includegraphics[width=1.2 \columnwidth]{bounds.pdf}
%     \caption{A plot of $\chi_{t}^{(general)}$ (red) and $\chi_{t}^{(bound)}$ (blue) versus $m$. It can be seen that there is a transition between polynomial scaling and exponential scaling at a critical point $m_{c} = \lfloor\frac{n}{2}\rfloor - 1$.}
%     \label{Fig: chivsm}
% \end{figure}

\subsection{General bounds for \texorpdfstring{$\chi_{t}(N,n,m)$}{TEXT}}
To begin bounding $\chi_{t}$ for a `matchgate + ZZ' circuit, we assume that we apply the 'Heisenberg picture' approach to matrix multiplication and that our measurement vector contains terms solely in $\mathcal{P}_{n}^{(2)}$. With these assumptions, the application of $m$ ZZ gates will at most access $m+1$ linear spans of even parity, up to a maximum of $\lfloor \frac{n}{2} \rfloor - 1$. The simplest bound of $\chi_{t}$ is therefore $N\chi_{max}$, where $\chi_{max}$ is the sum of the sizes of each accessed linear span. Explicitly: 

\begin{equation} \label{eqn: expression for chi_{max}}
\chi_{max} = \sum_{i=1}^{m+1}\vert L_{n}^{(2i)}\vert = \sum_{i=1}^{m+1}{2n \choose 2i}.
\end{equation} 
Such a partial sum has no closed form, and varies depending on whether $m$ is fixed as a function of $n$ or varies correspondingly. We consider both cases below. 

\subsubsection{Bounds for fixed \texorpdfstring{$m$}{TEXT}}

If $m$ is fixed in the above expression, we can show that the partial sum $\chi_{max}$ scales polynomially in $n$ as $n \rightarrow \infty$:

\begin{theorem}
The following relation holds:
\begin{equation}
\chi_{max} = \sum_{i=1}^{m+1}{2n\choose{2i}} < {2n\choose{2m+2}}\frac{1}{1-r},
\end{equation}
 where $r = (\frac{2m}{2n-2m-1})^{2}. $
% \begin{equation} \label{eqn: layeredbound1} \sum_{s=1}^{m}s{2n\choose{2m+ 4 -2s}} < {2n\choose{2m+2}}\frac{1}{(1-r)^{2}} \end{equation}
% where $r = (\frac{2m+2}{2n-2m-1})^{2}$, and $m \leq \lfloor\frac{n}{2}\rfloor -1$

\end{theorem}

\begin{proof}
Consider that for $k < \lfloor \frac{p}{2} \rfloor$: 
\begin{equation} \label{eqn: infinite series2}
 \frac{{p\choose{k}} + {p\choose{k-2}} + ...}{{p\choose{k}}} = 1 + \frac{k(k-1)}{(p - k +1)(p - k + 2)} + ...,
\end{equation}
where the numerator of the left hand side is the partial sum we wish to bound. Notice that we can upper bound the right hand side above by: 
\begin{equation} 
  1 + (\frac{k}{p-k+1})^{2} + ...,
\end{equation} 
which is an infinite geometric series with ratio $ r = (\frac{k}{p-k+1})^{2}$. This series converges to $\frac{1}{1-r}$, for $0 \leq r < 1$. Using the substitution that $p = 2n$ and $k = 2m + 2$, $r = (\frac{2m+2}{2n-2m-1})^{2}$ and the series converges as long as $0 \leq m \leq \lfloor\frac{n}{2}\rfloor - 1$. Rearranging equation \ref{eqn: infinite series2} completes the proof.
% To prove inequality 17, we can use the result that $\sum_{k=1}^{\infty} kr^{k}=  \frac{r}{(1-r)^{2}}$. If we correspondingly modify equation 19 with integer weights:  
% \begin{equation} 
%   1 + 2(\frac{k}{p-k+1})^{2} + ...,
% \end{equation} 
% we notice this is is equal to $\sum_{k=1}^{\infty} kr^{k-1} =  \frac{1}{(1-r)^{2}}$
\end{proof}

Let us consider the asymptotic scaling as $n \rightarrow \infty$ with fixed $m$. We can see that in this limit, $r \rightarrow 0$, and the scaling is dominated by ${2n\choose{2m+2}}$. Hence we can conclude the  following results, using the relation that ${2n\choose{2m+2}} < \frac{ne}{m+1}^{2m+2}$, where $e$ is Euler's constant: 

\begin{corollary}
    For fixed $m < m_{c}$, as $n \rightarrow \infty$, $\chi_{max}$ scales as $\mathcal{O}((\frac{ne}{m+1})^{2m+2})$.
\end{corollary}

% \begin{corollary}
%     For fixed $m < m_{c}$, as $n \rightarrow \infty$, $\chi_{bound}^{(layered)}$ scales as $\mathcal{O}(\frac{N}{m+1} (\frac{ne}{m+1})^{2m+2}).$
% \end{corollary}

This implies the total simulation cost of the algorithm, which is on the order of $\sim \mathcal{O}(\chi_{t})$ will in the worst case scale polynomially in $n$, for fixed $m$.

\subsection{Bounds for variable \texorpdfstring{$m$}{TEXT}}

We have considered fixed $m$ to this point. We now consider the scaling if $m$ increases as a function of $n$ in the limit where $n \rightarrow \infty$. We make use of the following bound, proved in \cite{MacWilliams1977TheTO}, for this purpose: 

\begin{lemma} 
The following relation holds for $0 \leq k \leq n$:
$$
S = \sum_{i=0}^{k}{2n \choose i} \leq  2^{2nH(\frac{k}{2n})}, 
$$

where $ H(x) = -x \log_{2}(x) - (1-x) \log_{2}(1-x)$ is the binary entropy function.
\end{lemma} 

As we are only concerned with even terms of this sum we can make the following modification. Split S into even and odd terms $S = S_{o} + S_{e}$, where $S_{o} = \sum_{i=1}^{m+1} {2n \choose 2i - 1}$ and $S_{e} = \sum_{i=0}^{m+1} {2n \choose 2i }$. Here, $S_{e}$ is the quantity we wish to bound. The ratio $\frac{S}{S_{e}} = 1 + \frac{S_{o}}{S_{e}} = 1 + O(n^{-1})$ to leading order, as the partial sum of . We can rearrange this expression to give a bound for $S_{e}:$

\begin{equation}
S_{e}  \leq \frac{1}{1 + O(n^{-1})}2^{2nH(\frac{k}{2n})}
\end{equation}

Finally, noting that $k = 2m+2$, We can bound $\chi_{max}$ in the asymptotic limit as follows: 

\begin{theorem}
For $m \leq \lfloor \frac{n}{2} \rfloor - 1:$
\begin{equation}
\chi_{max} < 2^{2nH(\frac{m+1}{n})}
\end{equation}
\end{theorem}

\subsection{Tighter bounds for structured circuits}
It is possible to give some tighter bounds on the simulation time complexity, if we assume some natural structure in the circuit. The motivation here is that the Pauli rank increases over the course of the circuit instead of being fixed at $\chi_{max}$. If the ZZ gates are equally spaced, we will have a staircase structure. Evaluating $\chi_{t}$ is then a question of splitting the staircase into horizontal steps and adding each section. In general the width of the kth step (starting from the bottom) each step will be $\frac{kN}{m+1}$, and the height will be $\vert L_{2k} \vert$. Hence we can approximate $\chi_{t}$ in the following way:

\begin{equation}  \label{eqn:partial sum}
\begin{aligned}
    \chi_{t}(n,m) \approx \\  
    & = \frac{N}{m+1}\sum_{i=1}^{m+1}(m+2-i)|L^{(2i)}_{n}| \\ 
    & = \frac{N}{m+1}\sum_{s=1}^{m}(m+2-i) {2n\choose{2i}},
\end{aligned}
\end{equation}
where $N = mT + m$, for some constant T which is the number of matchgates applied between each ZZ gate. 
This partial sum can be upper bounded by an arithmetico-geometric series as follows:

\begin{lemma}
\begin{equation}
\sum_{s=1}^{m+1}(m+2-i){2n\choose{2i}} < {2n\choose{2m+2}}[1+\frac{1}{(1-r)^{2}}],
\end{equation}
 where $r = (\frac{2m}{2n-2m+1})^{2}. $
% \begin{equation} \label{eqn: layeredbound1} \sum_{s=1}^{m}s{2n\choose{2m+ 4 -2s}} < {2n\choose{2m+2}}\frac{1}{(1-r)^{2}} \end{equation}
% where $r = (\frac{2m+2}{2n-2m-1})^{2}$, and $m \leq \lfloor\frac{n}{2}\rfloor -1$

\end{lemma}

\begin{proof}
Consider that for $k < \frac{p}{2}$: 
\begin{equation} \label{eqn: infinite series}
\begin{split}
 \frac{{p\choose{k}} + 2{p\choose{k-2}} + 3{p\choose{k-4}}+\dots}{{p\choose{k}}} \
& \\
= 1 + 2\frac{k(k-1)}{(p - k +1)(p - k + 2)} + \dotsb
\end{split}
\end{equation}
% \begin{equation} \label{eqn: infinite series}
%  \frac{{p\choose{k}} + 2{p\choose{k-2}} + 3{p\choose{k-4}}+...}{{p\choose{k}}} = 1 + \frac{k(k-1)}{(p - k +1)(p - k + 2)} + ...,
% \end{equation}
We notice that we can bound this above by: 
\begin{equation} 
S_{n} = 1 + 2(\frac{k}{p-k+1})^{2} + ...,
\end{equation} 
which is an infinite arithmetico-geometric series with ratio $ r = (\frac{k}{p-k+1})^{2}$:
\begin{equation}
S_{n} = 1 + \sum_{i = 1}^{\infty}(i+1)r^{i}
\end{equation}

The sum converges to $[1 + \frac{1}{(1-r)^{2}}]$, for $0 \leq r < 1$. Using the substitution that $p = 2n$ and $k = 2m + 2$, $r = (\frac{2m+2}{2n-2m-1})^{2}$ and the series converges as long as $0 \leq m \leq \lfloor\frac{n}{2}\rfloor - 1$. The numerator of the fraction in equation \ref{eqn: infinite series} is equal to the partial sum we wish to bound. Hence rearranging this expression gives us the desired result. 
% To prove inequality 17, we can use the result that $\sum_{k=1}^{\infty} kr^{k}=  \frac{r}{(1-r)^{2}}$. If we correspondingly modify equation 19 with integer weights:  
% \begin{equation} 
%   1 + 2(\frac{k}{p-k+1})^{2} + ...,
% \end{equation} 
% we notice this is is equal to $\sum_{k=1}^{\infty} kr^{k-1} =  \frac{1}{(1-r)^{2}}$
\end{proof}
We use this result in Theorem \ref{Theorem 6: Bound} to give a closed-form expression for a bound to $\chi_{max}$.

\begin{theorem}
\label{Theorem 6: Bound} For $m \leq m_{c}= \lfloor\frac{n}{2}\rfloor -1$:
\begin{equation}
    \chi_{max}(n,m) < \frac{N}{m+1}{2n\choose{2m+2}}[1+\frac{1}{(1-r)^{2}}],
\end{equation}

\end{theorem}
which means for structured circuits there is an approximately $\frac{1}{m}$ improvement over the loose bound derived in the previous section.

\section{Fermi-Hubbard Model Simulation using Pauli-basis Simulation}

\begin{figure*}[ht]
    \centering
    \includegraphics[width=\textwidth]{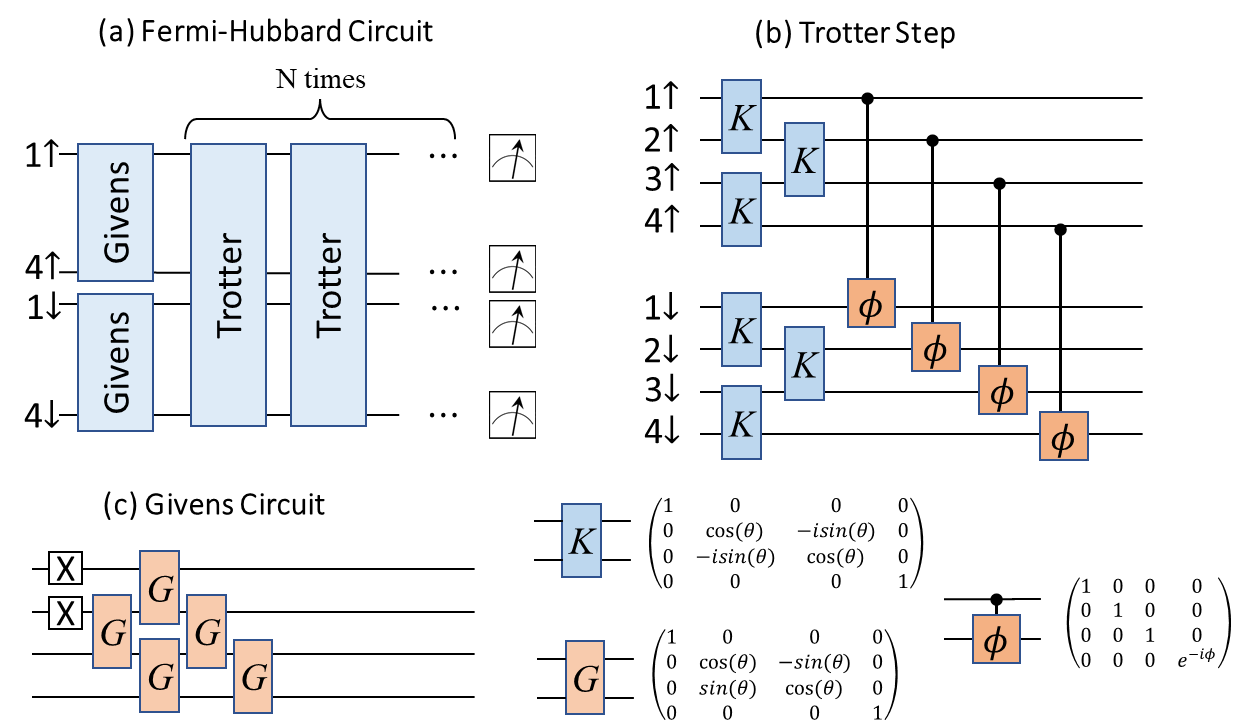}
    \caption{Circuit used to simulate time evolution of $H_{FH}$ given in Equation \ref{Equation: HFH}, from \cite{arute2020observation}. (a) Overall circuit, consisting of initialisation of Hartree-Fock state using Given's rotations as shown in (c). This is followed by $m$ Trotter steps, and finally, measurement (b) Trotter step, consisting of odd-even hopping gates, even-odd hopping gates (blue) and onsite interactions (orange). In the interaction-limited case, only one cPhase gate connects the top and bottom registers.}
  
\end{figure*}

\begin{figure*}[t] 
    \centering
    \hspace*{-1.5cm}
    \includegraphics[width = 1.2\textwidth]{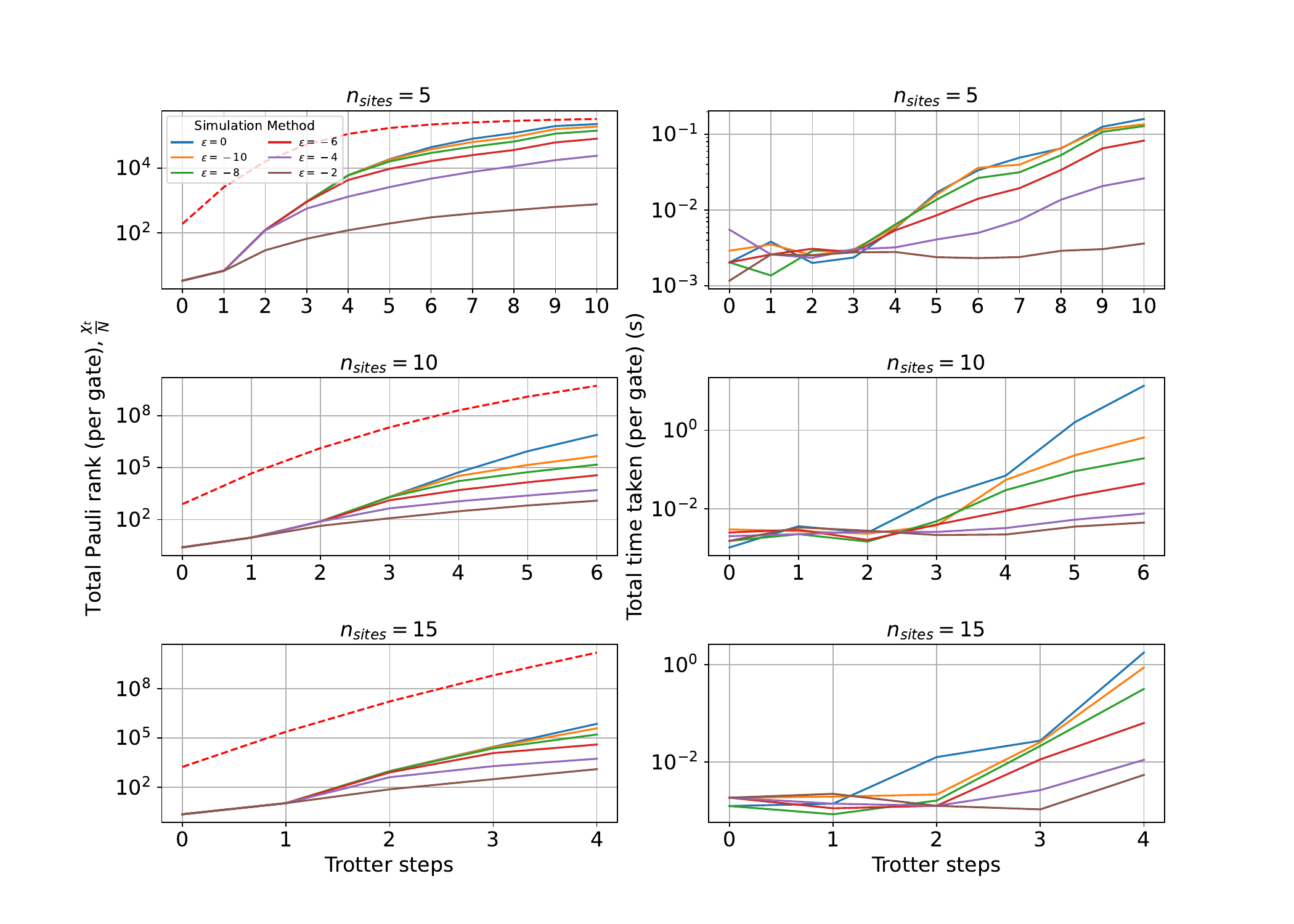}
    \caption{Total Pauli rank/Time (per gate) required to evaluate $\langle Z_{j} \rangle$ as a function of the number of Trotter Steps of an interaction limited version of the circuit shown in Figure 6. The red dashed line indicates the saturated upper bound for $\chi_{t}$ for the particular number of Trotter steps. For simulations with a pruning threshold, we can greatly reduce the simulation cost, however, we incur an error $\delta$ which varies as shown in Figure \ref{Figure: delta vs pauli rank}. }
\label{Figure: Simulation}
\end{figure*}

\begin{figure}[ht]
    \centering
    \hspace*{-0.5cm}
\includegraphics[width=1.1\columnwidth]{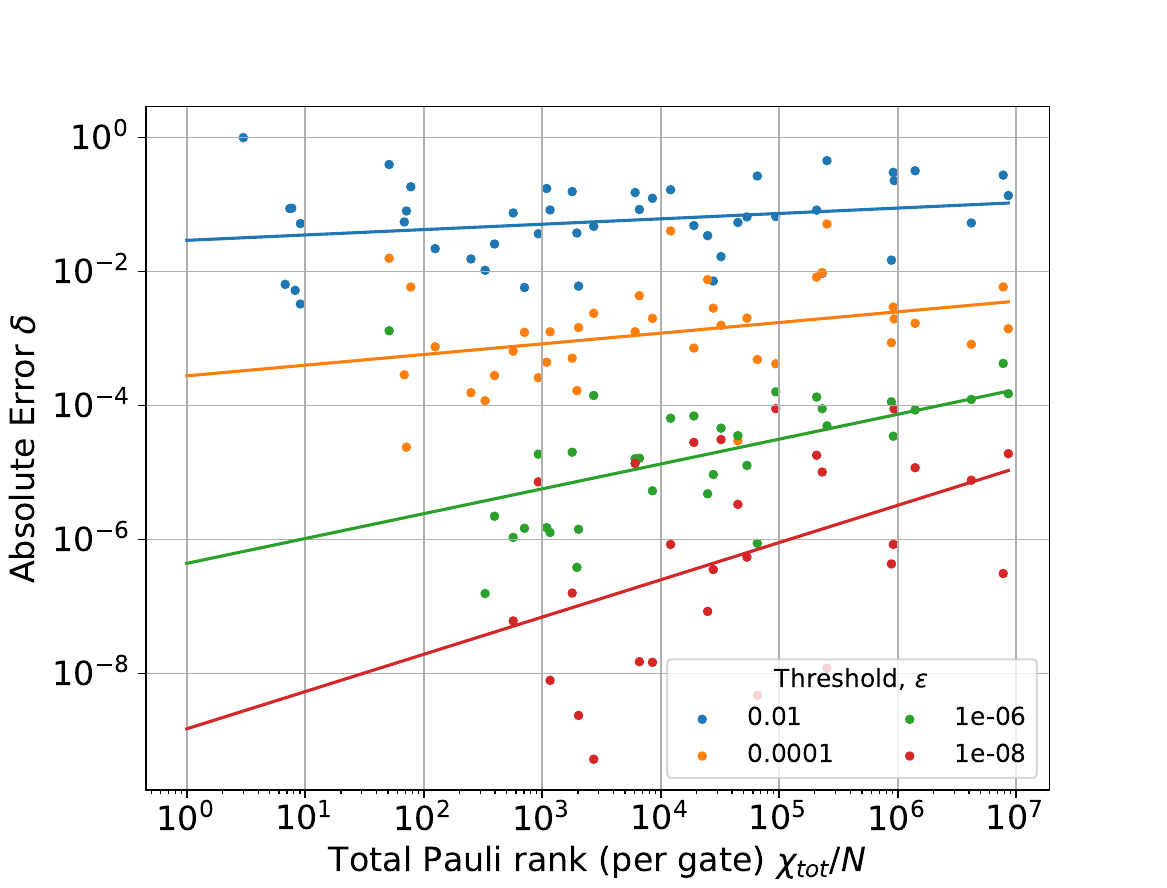}
    \caption{Absolute Error of expectation value for different values of pruning threshold, as a function of the total Pauli rank, $\chi_{tot}$ of a simulation. Each point corresponds to a Trotter circuit simulated with different values of $n$ and $T$. The same set of circuits is used for each threshold value, and where the absolute error was zero, no point is shown.}
    \label{Figure: delta vs pauli rank}
\end{figure}

In this section, we demonstrate how a 'matchgate + ZZ' circuit arising in the context of the Fermi-Hubbard model can be simulated using Algorithm 1, and verify the above asymptotic scalings. 
The Fermi-Hubbard model is a widely used model in condensed matter physics to understand the properties of correlated fermionic systems. A 1D model consists of $p$ sites, arranged in a line, with $q < 2p$ fermions distributed between these sites. Each site can contain a spin-up fermion, a spin-down fermion or both. Fermions can \textit{hop} between adjacent sites, characterised by an energy $J$, and interact with a coulombic repulsion each other if a site is doubly occupied, characterised by an energy $U$. The Hamiltonian of the system is given below:
\begin{equation} 
\begin{aligned} H_{FH}=&-J \sum_{j=1}^{L-1} \sum_{\nu=\uparrow, \downarrow} a_{j, \nu}^{\dagger} a_{j+1, \nu}+\text { h.c. } \\ &+U \sum_{j=1}^{L} n_{j, \uparrow} n_{j, \downarrow}, \end{aligned}
\label{Equation: HFH}
\end{equation}
where the first term corresponds to hopping and the second term to the onsite interaction terms respectively.

Simulating this Hamiltonian can be achieved as follows. Firstly, the initial state is prepared using a Givens rotation circuit as shown in Figure 4c. This consists of a ladder of matchgates of the form $G$, for each electron being simulated. In this case, we consider two electrons. The parameters of each gate are chosen randomly for our purposes, though in practice they are carefully chosen to correspond to a Hartree-Fock state. Next, the dynamics induced by $H_{FH}$ are approximated using the Lie product formula:
$$
e^{i(H_{hop} + H_{int})t} = \lim_{T\to\infty}\{e^{iH_{hop}\frac{t}{N}}e^{iH_{int}\frac{t}{T}}\}^{T},
$$
which approximates the dynamics as the product of the individual hopping and interacting unitary operators, applied for a fractional amount of time. Each individual unitary can be decomposed as a circuit on $2n$ qubits, where the first and second register of $n$ qubits correspond to spin up and spin down electrons respectively. The hopping terms are implemented as $K$ rotations (also matchgate circuits) on each register, and the interaction terms is implemented as CPhase gates acting between registers. This is shown in Figure 4b. This process is repeated for a particular number of steps known as the trotter number $T$. This gives an overall circuit structure as shown in Figure 4a. 

Writing the circuit explicitly, we can see that our circuit factorises into a form which is amenable to Algorithm 1:  $U_{FH} = U^{(T)}_{MG}U^{(T)}_{\phi}\dots U_{MG}^{(1)}U_{\phi}^{(1)}U^{(0)}_{MG}$, where $U_{MG}$ are matchgate circuits corresponding to the Givens and K rotations,  and $U_{\phi}$ are CPhase gates. Indeed, CPhase gates are matchgate equivalent to $U_{ZZ} = e^{i\theta ZZ}$, specifically taking the form $\textrm{CPhase}(\theta)$ = $e^{i \theta ZZ} G(R_{z}(2\theta),I)$.  The fact that these gates act on distant qubits a and b poses no issue (in this case) as $R_{z}^{a}\otimes I$ and $I \otimes R_{z}^{b}$ be expressed as G($R_{z}(\theta), R_{z}(\theta)$)) G($R_{z}(\theta), R_{z}(-\theta)$) respectively, and the transformations induced by $e^{i \theta Z_{a}Z_{b}}$ hold for any qubits a and b, according to Theorem \ref{Theorem: 4}. 

To test Algorithm 1, we consider the simulation of Trotter circuits with limited interactions, where only a finite number of CPhase are present in each Trotter step. This corresponds to removing on-site interaction for particular sites in the Fermi-Hubbard model. This paradigm may be useful for error mitigation methods which use efficient simulation of `training circuits' to understand and mitigate errors in a related universal circuit \cite{montanaro2021error}. Typically, such training circuits would have to be from gate sets which are classically simulable but may not be functionally similar to the circuit being implemented. By extending the reach of classical simulation methods to universal circuits it is possible that such methods could perform better.

In Figure \ref{Figure: Simulation} we show the scaling of $\chi_{tot}$ (per gate) as a function of the number of Trotter steps (blue line), each step containing a single $\textrm{CPhase}$ gate between the first qubit of the upper and lower register. In this case, the number of Trotter steps is equal to the number of parity-preserving non-matchgates. We simulate models consisting of $n_{sites} = 5,10,15$ sites (corresponding to $n$ = 10, 20, 30 qubits respectively), for a number of Trotter steps $T$ in the range [0,10]. We further normalise the rank and time by the number of gates, the number of which is given by $N = 2n_{sites}(T+1)$. From Figure 8, we verify our prediction that the simulator scales polynomially in the qubit number $n$, but exponentially in the number of parity-preserving non-matchgates. Furthermore, ignoring transient effects for small circuit depths, the time taken scales in proportion with the total Pauli rank.

\subsection{Heuristic Improvements}
We now mention two methods to improve the proposed simulation method. The first improvement is to use what is sometimes known as an \textit{interaction picture} approach, which describes applying R operators to both $\vert \rho_{0} \rangle \rangle$ and $\langle\langle Z_{j} \vert$.  This reduces the number of iterations with a large number of multiplications which often occur in the final gates of a layered circuit. It also introduces parallelism as the two directions of multiplication can be computed independently. To initialise  $\vert \rho_{0} \rangle \rangle$ using this method, one need only consider elements in the linear spans which will be accessed following the application of $m$ parity-preserving non-matchgates gates. For computational zero states, this will be Pauli operators containing up to $m$ $Z$ operators. 

The second method is to apply \textit{pruning}, where coefficients from $\langle\langle M\vert$ which are below some threshold value $\epsilon$ are removed before the application of the next super-operator. At the cost of some error in the final expectation value, this heuristic has the potential to greatly reduce the time cost for certain circuits [see figure 5].  The relationship between the incurred error $(\delta)$ and the total pruning-free ($\epsilon = 0$) Pauli rank (per gate) is plotted in Figure 6. We consider the range of ranks accessed in our simulations, for different values of the pruning threshold $(\epsilon)$. Each point corresponds to a randomly initialised Trotter circuit with a particular value of $n$ and $T$ in the ranges [4,7] and [0,8] respectively. The same circuit is simulated for each value of the threshold $\epsilon$.

\section{Discussion}

The motivating question for this research was whether the polynomial resource cost of nearest-neighbour matchgate circuit simulation techniques could be preserved for circuits supplemented with a few universality-enabling gates. Algorithm 1 shows that at the very least, polynomial scaling is still possible for up to $\lfloor\frac{n}{2}\rfloor - 1$ such gates. As could be expected,  the scaling quickly becomes exponential as more universality-enabling resources are added. It is however still possible that the exponent with respect to $m$ could be suppressed with more sophisticated simulation techniques. For example, one line of research could be the use of tensor networks in simulating universal circuits \cite{pan2021simulating}. In particular there are several schemes which efficiently encode fermionic gaussian states (ie states formed from matchgate circuits) as low-rank matrix-product states \cite{schuch2019matrix} \cite{fishman2015compression}. We make note that there have recently been approaches which make analogies to the stabilizer decomposition approach. We direct the reader to these papers. \cite{reardon2023improved}

We also mention here a few considerations for further investigation. Firstly, as the classical simulation results and scaling properties derived here stem from the fermionic description of operators, we can reason that using a different representation of these operators, such as the Bravyi-Kitaev, or a compact mapping for 2D fermionic grids \cite{derby2021compact}, will lead to equivalent simulation methods for different sets of gates. Albeit such gates will no longer necessarily be nearest-neighbour matchgates or the identified parity-preserving non-matchgates. A detailed exposition of such gate sets is an interesting question.

\section{Conclusion}
In this paper, we have shown that it is possible to extend the PI-SO matchgate simulation method introduced by Jozsa and Miyake \cite{Josza2008} to simulate a matchgate circuit containing a few parity-preserving non-matchgates (ZZ gates). At a high level, this is possible because of two observations. The first observation is that Gaussian operations induce linear transformation within a linear span $\mathcal{L}_{n}^{(k)}$, the vector space spanned by degree $k$ majorana monomials. The second observation is that non-Gaussian operations induce transformations across \textit{multiple} such linear spans. This gives a simple scheme to account for the addition of non-matchgates: adaptively increase the size of the vector space encompassing the dynamics over the course of a circuit. Using a sparse simulation method, we show that the time complexity of simulating a `matchgate + ZZ' depends on $N$, $n$, and $m$, the number of gates in total, the number of qubits and the number of $ZZ$ gates. For $m \leq \lfloor\frac{n}{2}\rfloor - 1$, we show the time cost scales as  $\sim \mathcal{O}(N(\frac{n}{m+1})^{2m+2})$ for fixed $m$ as $n \rightarrow \infty$. For variable $m$, however, we find a scaling of the form $2^{2nH(\frac{m+1}{2n}})$. Finally, We showcase this method in a simulation of Trotter circuits arising from simulation of an interaction-limited Fermi-Hubbard model. 

\section{Acknowledgements}
We would like to thank Daniel Brod for helpful discussions. We acknowledge funding from
the EPSRC Prosperity Partnership in Quantum Software for
Modelling and Simulation (Grant No. EP/S005021/1).

\bibliographystyle{unsrtnat}
\bibliography{bibliography.bib}
\onecolumn

\appendix

\section{Algorithm 1} \label{Algorithm 1}
In this section, we introduce an algorithm to evaluate an expression of the following form:

$$
\langle M \rangle = \langle\langle M \vert R \vert \rho_{0} \rangle\rangle,
$$
where $\langle \langle M \vert$ is a measurement operator vector, $R$ is a superoperator corresponding to conjugation by a circuit, and $\vert \rho_{0}\rangle \rangle$ is the initial density matrix vector. The algorithm can be shown to have a generic resource cost of $\mathcal{O}(\chi_{t})$, where $\chi_{t}$ is the total Pauli rank of the circuit, as defined in equation \ref{Total Pauli Rank}. This algorithm could in theory be applied to a computation with any measurement operator, initial state or circuit, however we consider in what follows matchgate + ZZ circuits simulated in the PI-SO setting. Such circuits take the generic form: 

\begin{equation}
U_{MG + ZZ} = U_{MG}^{(m)}U_{ZZ}^{(m)}...U_{MG}^{(1)}U_{ZZ}^{(1)}U_{MG}^{(0)},
\end{equation}
where there are $m + 1$ nearest-neighbour matchgate circuits $U_{MG}$ interspersed with $m$ $U_{ZZ}$ gates. We remind the reader that many non-matchgates of practical relevance, such as SWAP, CZ and CPhase are matchgate equivalent to $U_{ZZ}$, and the matchgate components of these gates can be virtually absorbed into the neighbouring matchgate circuits, though in practice this isn't necessary. Recasting in Liouville notation, we ca write the simulation problem explicitly as:  
\begin{equation} \label{eqn: expectation}
\langle Z_{j}\rangle = \langle\langle Z_{j} |  R_{MG}^{(m)}R_{ZZ}^{(m)}...R^{(1)}_{MG}R^{(1)}_{ZZ}R^{(0)}_{MG}| \rho_{0} \rangle\rangle,
\end{equation}

where the superoperator corresponding to each matchgate circuit is written as $R_{MG}$.  Furthermore, each $R^{(m)}_{ZZ}$ is the linear operator corresponding to a $U_{ZZ}$ gate acting on arbitrary nearest neighbour qubits. To evaluate equation \ref{eqn: expectation}  efficiently,  we  use the `Heisenberg picture' approach.. Furthermore, we use sparse data structures to ensure that no (initially) exponentially-sized object needs to be allocated to memory. Specifically, we can store $\langle \langle Z_{j}^{'} \vert$ during the computation as a dictionary of keys: $\langle \langle Z_{j}^{'} \vert = \{(P_{1}:v_{1}),\dots, (P_{\chi}:v_{\chi})\}$, where each key-value pair corresponds to a $n$-qubit Pauli operator and its coefficient in a Pauli-basis decomposition. Matrix-vector multiplication $ \langle \langle Z_{j}^{'} R\vert $ is then performed by (1), constructing the Pauli transfer matrix $O \in SO(16)$ for each gate, (2) cycling through each key in $ \langle \langle Z_{j}^{'} \vert$ and determining which basis elements are rotated by $O$, and (3) updating the corresponding values with matrix-vector multiplication. 

To do step (1) efficiently, we use the observation that each $R$ matrix is fully characterised by a rotation matrix $O \in SO(16)$ acting on a two-qubit subspace of Pauli operators, stemming from the fact that each gate is two-qubit. For each operator constituting an $R_{MG}$, the structure of $O$ is a sparse block-diagonal with 1,4,6,4,1 dimensional submatrices corresponding to a rotation of the basis of $P_{2}^{(k)}$. The submatrices are calculated as shown in Subroutine 1. For $R_{ZZ}$, the corresponding $O$ matrix can be calculated in constant time. For SWAP, for example, this will simply be a permutation matrix.

For step (2), we identify the elements of $\langle \langle Z_{j}^{'} \vert$ which are rotated by each $O$ matrix. To do this, we  split every key into its \textit{support} and a \textit{stem} on qubits $j,j+1$ denoted $P_{j,j+1}$ and $P_{n-j,j+1}$ respectively, describing the support of the key on the qubits $j$ and $j+1$, and its complement. The rotation will be the same for all keys with the same support. Hence, if $P_{j,j+1} \in P_{2}^{(d)}$, then the new basis elements following rotation are obtained by the tensor product of elements of  $P_{2}^{(d)}$ with the stem (Subroutine 2).

As an example, consider a rotation supported on qubits 1 and 2. Then a key $ZIXY$ is split into $ZI$ (support) and $XY$ (stem). As the support is in $P_{2}^{(2)}$, the basis elements accessed in a rotation will be $ZIXY, IZXY, XXXY, YYXY, XYXY, YXXY$. We can store dictionary keys as integers by encoding $n$-qubit Pauli operator as a bitstring of size $2n$, where $I = 00$, $X = 01$, $Y = 10$, and $Z = 11$. Find(P) can then be efficiently implemented as integer addition and subtraction. Finally, we perform an update of the identified key-value pairs by performing a matrix-vector multiplication between the multiplication of the coefficients $\vec{v}$ corresponding to $\vec{P}$ (Subroutine 3).

The overall matrix-vector multiplication uses these three subroutines as shown in  Algorithm 1. As the measurement vector is constantly being updated, it is important to keep track of which keys have already been rotated to avoid redundant multiplication. For this purpose, a separate data structure with $\mathcal{O}$(1) write/read cost, labelled \textit{temp} (such as a hash table) is used to store the keys already accessed using `Find'.

The total simulation cost is proportional to a quantity denoted $\chi_{t} = \sum_{i}\chi_{i}$ which is the total sum of the Pauli ranks at each step of the circuit. The relevance of $\chi_{t}$ follows from the fact that the  various contributions to the time cost are characterised by $\chi_{i} = \chi(Z_{j}^{'})$ at each iteration. Specifically, the number of `Find' and `Update' calls for each gate can be approximated by $\frac{\chi_{i}}{s} $, where $s$ is the maximum sparsity of each R matrix (which for $R_{MG}$ is s = 6, and $R_{ZZ}$ is s=1). Whereas `Find' can be implemented as integer addition (which is efficient), matrix-vector multiplications during `Update' calls will cost $\sim$ $\mathcal{O}(s^{2})$. Hence the dominant time cost for a single matrix-vector multiplication will be  $\mathcal{O}(\frac{\chi_{i}}{s}s^{2}) = \mathcal{O}(\chi_{i}s)$. Summing the costs across the entire circuit, the overall time cost will be $\mathcal{O}(\chi_{t}s)$. Hence, determining the overall time complexity of the algorithm reduces to finding suitable bounds for $\chi_{t}$. Finally, we note that evaluating the final inner product scales as $\mathcal{O}$ ($\chi_{max}$). 

\begin{figure}[ht]
\begin{minipage}{0.45\textwidth}
\centering
\fbox{
\begin{minipage}{\textwidth}
\underline{\textbf{Subroutine 1}: \textbf{Rotations(U)}} \newline
\textbf{Input:} $U$ \newline
\textbf{Output:} 1,4,6,4,1 dim matrices $R_{(0)},R_{(1)},R_{(2)},R_{(3)},R_{(4)}$\newline

For $d \in [0,4]:$

\hspace{1cm} For $p^{(d)}_{\alpha}$, $p^{(d)}_{\beta}$ in $P_{2}^{(d)}$:

\hspace{2cm} $[R_{(d)}]_{\alpha \beta} = \mathrm{Tr}(U^{\dagger}p_{\alpha}^{(d)}Up_{\beta}^{(d)})$ 
\end{minipage}}
\end{minipage}
\hfill
\begin{minipage}{0.45\textwidth}
\centering
\fbox{
\begin{minipage}{\textwidth}
\underline{\textbf{Subroutine 2: Find(P)}} \newline
\textbf{Input:} An n-qubit Pauli operator $P$, qubit indices $j,j+1$ \newline
\textbf{Output:} Subset of $P_{2}$, $d$ \newline
\hspace*{1.5cm} $\vec{P} = \{P_{1},...,P_{m}\}$ where $m = \vert P_{2}^{(d)} \vert$,\newline
$\textit{support} \gets  P_{j,j+1} \in P_{2}^{(d)} $ \newline
$\textit{stem} \gets  P_{n - (j,j+1)}$  \newline

For $p^{(d)} \in P_{2}^{(d)}$: \newline
\hspace*{1cm} $P_{\alpha}  =  P_{n - (j,j+1)} \otimes p^{(d)}$
\end{minipage}}
\end{minipage}
\end{figure}

\vspace{-0.5cm}

\begin{figure}[ht]
\centering
\fbox{
\begin{minipage}{0.5\textwidth}
\underline{\textbf{Subroutine 3: Update($\langle\langle M \vert, \vec{P}, R^{i}_{(d)})$)}} \newline
\textbf{Input:} A set of keys $\vec{P}$, DofK $\langle\langle M \vert $ \newline
\textbf{Output:} Updated DofK $\langle\langle M \vert $ \newline

For $P_{i} \in \vec{P}$: \newline
\hspace*{1cm} If $(P_{i}, v_{i}) \in \langle\langle M \vert$: 
   $\vec{v} \gets v_{i}$ \newline
\hspace*{1cm}  Else: $\vec{v} \gets 0$ \newline
$\vec{v^{'}} \gets R_{(d)}\vec{v}$ \newline
$\langle\langle M \vert \gets ({\vec{P}}, \vec{v})$
\end{minipage}}
\end{figure}

\begin{figure*}[t]
\begin{center}
\fbox{
\begin{minipage}{\textwidth}
\underline{\textbf{Algorithm 1}: Sparse Pauli-basis simulation for 'Matchgate+ZZ' circuit } \newline
\textbf{Inputs:} Observable M, 'matchgate + ZZ' circuit: $U = U^{(m)}_{MG}U^{(m)}_{ZZ}...U_{MG}^{(1)}U_{ZZ}^{(1)}U_{MG}^{(0)}$ on $n$ qubits, Initial product state $\rho_{0}$.
                  
\textbf{Outputs:} Expectation value $\langle M \rangle = Tr(U^{\dagger}MU \rho_{0})$

\textbf{Procedure:} \newline
Initialise measurement vector (e.g., $\langle \langle M \vert $ ${\gets (Z_{k}:1.0)}$) as a dictionary of keys.

\underline{Matrix-vector Multiplication}

For gates $U^{(i)}_{j,j+1} \in U_{MG}$ acting on qubits $j$ and $j+1$:
\hspace*{0.9cm} $R^{(i)}_{(0)},R^{(i)}_{(1)},R^{(i)}_{(2)},R^{(i)}_{(3)},R^{(i)}_{(4)} \gets \textbf{Rotations}(U^{(i)}_{j,j+1})$ \newline
\hspace*{1.1cm} $\textit{temp} \gets \{\}$ \newline
\hspace*{1cm} For $(P^{(m)},v^{(m)}) \in \langle\langle M \vert $: \newline
\hspace*{2cm} If $P^{(m)} \notin \textit{temp}$: \newline
\hspace*{3cm} $\vec{P}, d \gets$ \textbf{Find($P^{(m)}$)}\newline
\hspace*{3cm} $\langle\langle M\vert \gets \textbf{Update}(\langle\langle M \vert, \vec{P}, R^{(i)}_{(d)})$ \newline
\hspace*{3cm} $\textit{temp} \gets \vec{P}$

\underline{Inner Product} \newline
For $(P_{n}^{(m)},v_{m}) \in \langle\langle M \vert$:
\hspace*{1cm} $\langle{M}\rangle \mathrel{+}= Tr(\rho_{0}P_{n}^{(m)})v_{m}$
\end{minipage}}
\end{center}
\end{figure*}

\end{document}